\documentclass[11pt]{article}
\usepackage{wrapfig}
\pdfoutput=1

\usepackage{libertine}

\usepackage{geometry}
\usepackage{enumitem}
\PassOptionsToPackage{vlined, boxruled}{algorithm2e}
\usepackage{algorithm2e}
\usepackage{mleftright}
\usepackage{relsize}

\usepackage{amsmath}
\usepackage{amsthm}
\usepackage{amssymb}
\usepackage{graphicx}
\usepackage{tabularx}
\newcolumntype{Y}{>{\centering\arraybackslash}X}

\usepackage{array}
\usepackage{soul}
\usepackage[dvipsnames]{xcolor}
\usepackage{mdframed}
\usepackage{xparse}
\usepackage{hyperref}
\usepackage[capitalise, noabbrev]{cleveref}

\usepackage{tcolorbox}
\usepackage{tikz,lipsum}
\usepackage{newtxmath}
\tcbuselibrary{skins,breakable}

\makeatletter

\theoremstyle{plain}
\newtheorem{thm1}{Theorem}[section]
\theoremstyle{remark}
\newtheorem{pthm1}[thm1]{Theorem}
\theoremstyle{plain}
\newtheorem{lem1}[thm1]{Lemma}
\theoremstyle{plain}
\newtheorem{obs1}[thm1]{Observation}
\theoremstyle{plain}
\newtheorem{inv1}[thm1]{Invariant}
\theoremstyle{plain}
\newtheorem{cor1}[thm1]{Corollary}
\theoremstyle{definition}
\newtheorem{defn1}[thm1]{Definition}
\theoremstyle{plain}
\newtheorem{fact1}[thm1]{Fact}
\theoremstyle{remark}
\newtheorem{rem1}[thm1]{Remark}
\theoremstyle{plain}
\newtheorem{prop1}[thm1]{Proposition}
\theoremstyle{plain}
\newtheorem{asmp1}[thm1]{Assumption}

\crefname{thm1}{Theorem}{Theorems}
\crefname{pthm1}{Theorem}{Theorems}
\crefname{lem1}{Lemma}{Lemmas}
\crefname{obs1}{Observation}{Observations}
\crefname{inv1}{Invariant}{Invariants}
\crefname{cor1}{Corollary}{Corollaries}
\crefname{defn1}{Definition}{Definitions}
\crefname{fact1}{Fact}{Facts}
\crefname{rem1}{Remark}{Remarks}
\crefname{prop1}{Proposition}{Propositions}
\crefname{asmp1}{Assumption}{Assumptions}
\crefname{subsection}{Subsection}{Subsections}
\crefname{figure}{Figure}{Figures}

\ifx\proof\undefined
\newenvironment{proof}[1][\protect\proofname]{\par
\normalfont\topsep6\p@\@plus6\p@\relax
\trivlist
\itemindent\parindent
\item[\hskip\labelsep\scshape #1]\ignorespaces
}{%
\endtrivlist\@endpefalse
}
\providecommand{\proofname}{Proof}
\fi

\def\special{1}
\def\specialproof{1}
\def\specialdefinition{1}
\def\specialremark{1}
\def\highlight{0}
\def\sidebar{1}

\def\colorequations{0}

\newenvironment{thm}[1][]{%
\begin{thm1}[#1]%
}{\end{thm1}%
}
\newenvironment{lem}[1][]{%
\begin{lem1}[#1]%
}{\end{lem1}%
}
\newenvironment{obs}[1][]{%
\begin{obs1}[#1]%
}{\end{obs1}%
}
\newenvironment{inv}[1][]{%
\begin{inv1}[#1]%
}{\end{inv1}%
}
\newenvironment{fact}[1][]{%
\begin{fact1}[#1]%
}{\end{fact1}%
}
\newenvironment{rem}[1][]{%
\begin{rem1}[#1]%
}{\end{rem1}%
}
\newenvironment{pthm}[1][]{%
\begin{pthm1}[#1]%
}{\end{pthm1}%
}
\newenvironment{cor}[1][]{%
\begin{cor1}[#1]%
}{\end{cor1}%
}
\newenvironment{defn}[1][]{%
\begin{defn1}[#1]%
}{\end{defn1}%
}
\newenvironment{asmp}[1][]{%
\begin{asmp1}[#1]%
}{\end{asmp1}%
}
\newenvironment{prop}[1][]{%
\begin{prop1}[#1]%
}{\end{prop1}
}%

\if 1
    \if 1\highlight
        \renewenvironment{thm}[1][]{%
        \begin{mdframed}[nobreak=false,backgroundcolor=Aquamarine!60]\begin{thm1}[#1]%
        }{\end{thm1}\end{mdframed}%
        }
        \renewenvironment{lem}[1][]{%
        \begin{mdframed}[nobreak=false,backgroundcolor=YellowGreen!60]\begin{lem1}[#1]%
        }{\end{lem1}\end{mdframed}%
        }
        \renewenvironment{obs}[1][]{%
        \begin{mdframed}[nobreak=false,backgroundcolor=Salmon!60]\begin{obs1}[#1]%
        }{\end{obs1}\end{mdframed}%
        }

        \renewenvironment{prop}[1][]{%
        \begin{mdframed}[backgroundcolor=Goldenrod!60]\begin{prop1}[#1]%
        }{\end{prop1}\end{mdframed}%
        }
        
    \fi
    \if 1\sidebar
        \tcolorboxenvironment{thm}{blanker,breakable,left=4mm, before skip=10pt,after skip=10pt, borderline west={2mm}{0pt}{Aquamarine}}
        \tcolorboxenvironment{lem}{blanker,breakable,left=4mm, before skip=10pt,after skip=10pt, borderline west={2mm}{0pt}{YellowGreen}}
        \tcolorboxenvironment{obs}{blanker,breakable,left=4mm, before skip=10pt,after skip=10pt, borderline west={2mm}{0pt}{Salmon}}
        \tcolorboxenvironment{inv}{blanker,breakable,left=4mm, before skip=10pt,after skip=10pt, borderline west={2mm}{0pt}{Salmon}}
        \tcolorboxenvironment{fact}{blanker,breakable,left=4mm, before skip=10pt,after skip=10pt, borderline west={2mm}{0pt}{Salmon}}
        \tcolorboxenvironment{pthm}{blanker,breakable,left=4mm, before skip=10pt,after skip=10pt, borderline west={2mm}{0pt}{Salmon}}
        \tcolorboxenvironment{prop}{blanker,breakable,left=4mm, before skip=10pt,after skip=10pt, borderline west={2mm}{0pt}{Goldenrod}}
        \tcolorboxenvironment{cor}{blanker,breakable,left=4mm, before skip=10pt,after skip=10pt, borderline west={2mm}{0pt}{Mulberry}}
        \tcolorboxenvironment{asmp}{blanker,breakable,left=4mm, before skip=10pt,after skip=10pt, borderline west={2mm}{0pt}{Goldenrod}}
    \fi
\fi

\if 1\specialproof
    \if 1\highlight
        \expandafter\let\expandafter\oldproof\csname\string\proof\endcsname
        \let\oldendproof\endproof
        \renewenvironment{proof}[1][\proofname]{%
        \begin{mdframed}[nobreak=false,backgroundcolor=lightgray!60]\oldproof[#1]%
        }{\oldendproof\end{mdframed}}
    \else
        \if 1\sidebar
        \tcolorboxenvironment{proof}{
        blanker,breakable,left=4mm,
        before skip=10pt,after skip=10pt,
        borderline west={2mm}{0pt}{lightgray}}
        \fi
    \fi
\fi

\if 1\specialdefinition
    \if 1\highlight
        \renewenvironment{defn}[1][]{%
        \begin{mdframed}[innerbottommargin=0.1cm,innertopmargin=0.1cm,backgroundcolor=Apricot!60]\begin{defn1}[#1]%
        }{\end{defn1}\end{mdframed}%
        }
    \else
        \if 1\sidebar
        \tcolorboxenvironment{defn}{
        blanker,breakable,left=4mm,
        before skip=10pt,after skip=10pt,
        borderline west={2mm}{0pt}{Apricot}}
        \fi
    \fi
\fi

\if 1\specialremark
    \if 1\highlight
        
    \else
        \if 1\sidebar
        \tcolorboxenvironment{rem}{
        blanker,breakable,left=4mm,
        before skip=10pt,after skip=10pt,
        borderline west={2mm}{0pt}{Salmon}}
        \fi
    \fi
\fi

\if 1\colorequations
    \everymath{\color{MidnightBlue}}
\fi

\@ifundefined{showcaptionsetup}{}{%
    \PassOptionsToPackage{caption=false}{subfig}}
\usepackage{subfig}

\usepackage{float}
\floatstyle{boxed}
\restylefloat{figure}

\SetKwInOut{Input}{Input}
\SetKwProg{Initialization}{Initialization}{}{}
\SetKwProg{Fn}{Function}{}{end}
\SetKwProg{EFn}{Event Function}{}{end}
\SetKwFunction{CreateInstance}{CreateInstance}
\SetKw{And}{and}
\SetKw{Or}{or}
\SetKw{Break}{break}
\SetKw{Continue}{continue}

\SetCommentSty{mycommfont}

\SetFuncSty{myfuncsty}

\LinesNumbered \RestyleAlgo{boxruled}

\date{}

\makeatother

\makeatletter
\newcommand{\algorithmfootnote}[2][\footnotesize]{%
    \let\old@algocf@finish\@algocf@finish
    \def\@algocf@finish{\old@algocf@finish
    \leavevmode\rlap{\begin{minipage}{\linewidth}
                         #1#2
    \end{minipage}}%
    }%
}
\makeatother

\definecolor{darkred}{RGB}{200,0,0}

\newcommand{\opt}{\operatorname{OPT}}

\newcommand{\alg}{\operatorname{ALG}}

\newcommand{\pr}[1]{\mleft(#1\mright)}

\newcommand{\pc}[1]{\mleft\{#1\mright\}}

\newcommand{\ps}[1]{\mleft|#1\mright|}

\newcommand{\ceil}[1]{\mleft\lceil#1\mright\rceil}

\newcommand{\floor}[1]{\mleft\lfloor#1\mright\rfloor}

\newcommand{\I}[2]{\mleft(#1,#2\mright]}

\newcommand{\IR}[2]{\mleft[#1,#2\mright)}

\newcommand{\req}{q}

\newcommand{\cls}[1]{\ell_{#1}}

\newcommand{\pt}[1]{p_{#1}}

\NewDocumentCommand{\Wtgen}{s!D[]{}!D<>{}!d()}{
    \IfBooleanTF{#1}{
        \IfNoValueTF{#4}{
            W^{*#3}_{#2}
        }{
            W^{*#3}_{#2}\pr{#4}
        }
    }{
        \IfNoValueTF{#4}{
            W^{#3}_{#2}
        }{
            W^{#3}_{#2}\pr{#4}
        }
    }
}

\NewDocumentCommand{\sgen}{s!D[]{}!D<>{}!d()}{
    \IfBooleanTF{#1}{
        \IfNoValueTF{#4}{
            \delta^{*#3}_{#2}
        }{
            \delta^{*#3}_{#2}\pr{#4}
        }
    }{
        \IfNoValueTF{#4}{
            \delta^{#3}_{#2}
        }{
            \delta^{#3}_{#2}\pr{#4}
        }
    }
}

\NewDocumentCommand{\Vgen}{s!D[]{}!D<>{}!d()}{
    \IfBooleanTF{#1}{
        \IfNoValueTF{#4}{
            V^{*#3}_{#2}
        }{
            V^{*#3}_{#2}\pr{#4}
        }
    }{
        \IfNoValueTF{#4}{
            V^{#3}_{#2}
        }{
            V^{#3}_{#2}\pr{#4}
        }
    }
}

\NewDocumentCommand{\reqsgen}{s!D[]{}!D<>{}!d()}{
    \IfBooleanTF{#1}{
        \IfNoValueTF{#4}{
            Q^{*#3}_{#2}
        }{
            Q^{*#3}_{#2}\pr{#4}
        }
    }{
        \IfNoValueTF{#4}{
            Q^{#3}_{#2}
        }{
            Q^{#3}_{#2}\pr{#4}
        }
    }
}

\NewDocumentCommand{\Wt}{sm}{
    \IfBooleanTF{#1}{
        \Wtgen<*>(#2)
    }{
        \Wtgen(#2)
    }
}

\newcommand{\ept}[1]{\tilde{p}_{#1}}

\NewDocumentCommand{\xpt}{om}{\IfNoValueTF{#1}{x_{#2}}{x_{#2}(#1)}}

\NewDocumentCommand{\score}{om}{\IfNoValueTF{#1}{\Lambda_{#2}}{\Lambda_{#2}({#1})}}

\newcommand{\sept}{\operatorname{\mathbf{SEPT}}}

\newcommand{\sr}{\operatorname{\mathbf{SR}}}

\newcommand{\zig}{\textup{\textsc{\color{RoyalBlue}zig}}}

\newcommand{\zag}{\textup{\textsc{\color{BrickRed}zag}}}

\newcommand{\zigzag}{\textup{\textsc{\color{RoyalPurple}zigzag}}}

\newcommand{\zza}{\operatorname{\mathbf{ZigZag}}}

\newcommand{\thres}{\operatorname{\mathbf{DL}}}

\newcommand{\rlt}[1]{r_{#1}}

\NewDocumentCommand{\cov}{omm}{\IfNoValueTF{#1}{B\pr{#2,#3}}{B_{#1}\pr{#2,#3}}}

\newcommand{\lastt}[1]{t_{#1}}

\newcommand{\dstr}{\mu}

\newcommand{\edstr}{\hat{\dstr}}

\newcommand{\udstr}{\dstr_1}

\newcommand{\ddstr}{\dstr_2}

\newcommand{\sep}{\sigma}

\newcommand{\esep}{\hat{\sigma}}

\begin{document}
    \title{Distortion-Oblivious Algorithms for Minimizing Flow Time}
    \author{%
    \centering\begin{tabularx}{\textwidth}{YYY}
        \begin{tabular}[width = 0.3\textwidth]{c}
            Yossi Azar\tabularnewline
            \textsf{\small{}azar@tau.ac.il}\tabularnewline
            {\small{}Tel Aviv University}\tabularnewline
        \end{tabular}&
        \begin{tabular}[width = 0.3\textwidth]{c}
            Stefano Leonardi\tabularnewline
            \textsf{\small{}leonardi@diag.uniroma1.it}\tabularnewline
            {\small{}Sapienza University of Rome}\tabularnewline
        \end{tabular}&
        \begin{tabular}[width = 0.3\textwidth]{c}
            Noam Touitou\tabularnewline
            \textsf{\small noam.touitou@cs.tau.ac.il}\tabularnewline
            {\small{}Tel Aviv University}\tabularnewline
        \end{tabular}
    \end{tabularx}
    }
    \maketitle

    \begin{abstract}

We consider the classic online problem of scheduling on a single machine to minimize total flow time.
In STOC 2021, the concept of robustness to distortion in processing times was introduced: for every distortion factor $\dstr$, an $O(\dstr^2)$-competitive algorithm $\alg_{\dstr}$ which handles distortions up to $\dstr$ was presented.
However, using that result requires one to know the distortion of the input in advance, which is impractical.

We present the first \emph{distortion-oblivious} algorithms: algorithms which are competitive for \emph{every} input of \emph{every} distortion, and thus do not require knowledge of the distortion in advance.
Moreover, the competitive ratios of our algorithms are $\tilde{O}(\dstr)$, which is a quadratic improvement over the algorithm from STOC 2021, and is nearly optimal (we show a randomized lower bound of $\Omega(\dstr)$ on competitiveness).
    \end{abstract}

    \section{Introduction}
    \label{sec:Intro}
    We consider the classic online scheduling problem on a single machine.
Jobs arrive over time, and the completion of each job requires some processing time on the machine.
The goal of the algorithm is to minimize the \emph{total flow time}, which is the sum over jobs of their time pending (i.e. from release to completion).
This is an online problem, as the algorithm must decide which job to process at any given time, and is unaware of jobs that have not yet been released.

\paragraph{The classic setting.}
In the classic setting, the processing time of a job becomes known to the algorithm upon the job's release.
This knowledge gives great power to the algorithm: a classic result from the 1950s by Smith \cite{NAV:NAV3800030106} shows that the SRPT (shortest remaining processing time) algorithm is 1-competitive for this problem.

However, the assumption that processing times can be exactly known does not usually hold in practice.
For example, if a job is a computer program, the cases in which the running time is known exactly in advance are very rare.

\paragraph{The robust model.}
This lack of knowledge regarding processing times is addressed by the \emph{robust scheduling} model, presented in~\cite{DBLP:journals/corr/abs-2103-05604}.
In this model, upon the release of a job, the algorithm is provided an estimate for the processing time of the job.
Such estimates could be realistically obtained through heuristics, machine-learning models, or user input.
Naturally, one would expect the competitive ratio of an algorithm to improve with the accuracy of the provided estimations.

We define the \emph{overestimation} of an input to be the maximum estimated-to-real ratio of a job's processing time.
Similarly, we define the \emph{underestimation} to be the maximum real-to-estimated ratio of a job's processing time.
Finally, the \emph{distortion} of an input, denoted by $\dstr$, is the product of the overestimation and the underestimation of that input\footnote{The maximum ratios (and not e.g. average ratios) are indeed the correct parameters to consider; see~\cite{DBLP:journals/corr/abs-2103-05604}.}.

One desires a competitive ratio that is a function of this distortion $\dstr$.
Formally, for every distortion cap $\edstr > 1$, an $O(f(\edstr))$-competitive algorithm is \emph{$\edstr$-robust} if it remains $O(f(\edstr))$-competitive for all  inputs in which $\dstr \le \edstr$ (i.e. the distortion does not exceed the distortion cap).

In~\cite{DBLP:journals/corr/abs-2103-05604}, for every distortion cap $\edstr$, a $\edstr$-robust, $O(\edstr^2)$-competitive algorithm $\alg_{\edstr}$ is presented.
A major drawback of this result is that the algorithms for different values of $\edstr$ are different, which limits their usability.
The two options for using these algorithms are:
\begin{itemize}
    \item If one knows $\dstr$ in advance, one can choose $\edstr=\dstr$ and use $\alg_{\dstr}$.
    \item One can guess a value for $\edstr$ and hope that it is close to $\dstr$.
\end{itemize}
The first option is problematic, as knowing the distortion exactly in advance is rather unreasonable.
The second option is quite perilous; we show in \cref*{sec:BadCases} the poor consequences of choosing the distortion cap $\edstr$ inaccurately:
\begin{enumerate}
    \item When $\dstr \ll \edstr$, the algorithm $\alg_{\edstr}$ is $\Omega(\edstr)$-competitive, i.e. its competitiveness is not a function of the actual distortion $\dstr$ but of the distortion cap $\edstr$ (this happens even when there is \emph{no} distortion!).
    \item When $\dstr = 4\edstr$, i.e. when the guess is slightly too low, the algorithm $\alg_{\edstr}$ has unbounded competitiveness.
\end{enumerate}

\paragraph{Distortion-oblivious algorithms.}
Based on the previous discussion on robust algorithms, we desire a stronger guarantee than simple robustness: we would like a \emph{single} algorithm which works well for every input, with a competitive ratio that is tailored to the input's distortion $\dstr$ (rather than a distortion cap $\edstr$).
Formally, the guarantee provided by robustness is

\[
    \boxed{
    \forall \dstr >1: \exists \alg_{\dstr}: \text{$\alg_{\dstr}$ is $O(f(\dstr))$-competitive on inputs with distortion $\dstr$}
    }
\]
whereas the stronger guarantee would be
\[
    \boxed{
    \exists \alg: \forall \dstr >1:\text{$\alg$ is $O(f(\dstr))$-competitive on inputs with distortion $\dstr$}
    }
\]
We call an algorithm with this stronger guarantee a \textbf{distortion-oblivious} algorithm.
To conclude, the advantages of a distortion-oblivious algorithm $\alg$ over the robust algorithms $\pc{\alg_{\edstr}}_{\edstr>1}$ of~\cite{DBLP:journals/corr/abs-2103-05604} are:
\begin{enumerate}
    \item When $\dstr \ll \edstr$, $\alg$ would be $O(f(\dstr))$-competitive while $\alg_{\edstr}$ would be $\Omega(\edstr)$-competitive.
    \item When $\dstr > \edstr$, $\alg$ would remain $O(f(\dstr))$-competitive while $\alg_{\edstr}$ would break down completely.
\end{enumerate}

\subsection{Our Results}
In this paper, we present the first distortion-oblivious algorithms for scheduling with the goal of minimizing total flow time.

\begin{enumerate}
    \item We present the deterministic, distortion-oblivious $\zza$ algorithm.
    For every $\dstr$, $\zza$ is $O(\dstr\log \dstr)$-competitive for all inputs of distortion $\dstr$ (with both underestimations and overestimations).
    \item We show an $\Omega(\dstr)$ lower bound on competitiveness for every randomized algorithm for inputs with distortion of $\dstr$.
    This implies that the our two algorithms have a nearly-optimal competitive ratio.
\end{enumerate}

A salient feature of the $\zza$ algorithm is that it manages to completely ignore the distortion of jobs observed through processing.
An alternative approach could be to design an algorithm which attempts to learn the distortion of the input from those observations; using this technique we obtain the following algorithm.

\begin{enumerate}
    \setcounter{enumi}{2}
    \item We present the deterministic, distortion-oblivious algorithm $\thres$.
    For every $\dstr$, $\thres$ is $O(\dstr\log^2 \dstr)$-competitive for inputs of distortion $\dstr$ in which processing times are only underestimated (and never overestimated).
\end{enumerate}

Note that this algorithm is unable to handle overestimations, and is thus less general than $\zza$; this is an inherent result of trying to learn from observed distortions, which we discuss further in \cref{subsec:Intro_OurTechniques}.

The $\zza$ algorithm is superior to $\thres$ in both competitive ratio and generality (as $\zza$ supports both underestimations and overestimations); the paper thus focuses on presenting and analyzing $\zza$.
We still present and analyze $\thres$ in the appendix, as we believe that the relevant techniques could be of independent interest.

The nearly-linear competitive ratios of our algorithms improves upon the best previously-known result in~\cite{DBLP:journals/corr/abs-2103-05604}, which requires knowing $\dstr$ in advance (non-oblivious) and even then is $O(\dstr^2)$-competitive.
As our lower bound shows, our algorithms are optimal up to logarithmic factors.

\subsection{Our Techniques}
\label{subsec:Intro_OurTechniques}
Both $\thres$ and $\zza$ rely on the two scheduling strategies $\sept$ and $\sr$.
On their own, each of these strategies has an unbounded competitive ratio even in the presence of constant distortion; however, each strategy has a different merit which $\thres$ and $\zza$ can exploit.

The first strategy is $\sept$ -- scheduling according to Shortest Estimated Processing Time.
In this strategy, we divide jobs into exponential classes according to estimated processing time, and always work on a job from the smallest class (i.e. shortest estimate).

While this strategy has unbounded competitive ratio even without the presence of distortion (see \cref{sec:BadCases}), it has the merit of always working on a job of the lowest class -- that is, the most cost-effective job available.

The second strategy is $\sr$ (``special rule''), presented in~\cite{DBLP:journals/tcs/BecchettiLMP04}.
The strategy distinguishes between jobs that have not been processed at all (`full') and jobs that have (`partial').
This strategy always works on the partial job of the minimal class, unless there are at least two full jobs of lower class: in that case, the one with the lowest class is marked as partial.

The main merit of this strategy is that it maintains that the number of partial jobs is at most a constant time the number of full jobs.
This ensures that the algorithm never reaches a bad state in which almost all jobs are partial and nearly-complete; in such a state, the optimal solution could finish those nearly-complete jobs and be well ahead of the algorithm.
However, as previously stated, the $\sr$ strategy fails to be robust: \cref{sec:BadCases} shows that its competitive ratio is unbounded even in the presence of constant distortion ($\dstr=4$).

\paragraph{The $\thres$ algorithm.}
Consider the case in which a full job $r$ is of a smaller class than the minimum-class partial job $q$.
If the algorithm requires a second full job of class smaller than $q$ in order to switch to $r$, it follows the $\sr$ strategy.
If the algorithm allows the second full job to be of \emph{any} class, this is very similar to simply following the $\sept$ strategy.
It turns out that a middle ground that leads to robustness is to require a full job which is up to $\log \dstr$ classes above $q$.
The $\thres$ algorithm attempts to implement this rule, but it does not know $\dstr$; instead, $\thres$ uses the maximum distortion seen thus far.

This turns out to be sufficient when there are only underestimations (which is when $\thres$ is $O(\dstr\log^2\dstr)$-competitive), but fails in the case of overestimations.
An intuitive explanation for this is that when the adversary is given the ability to underestimate and overestimate jobs, the underestimations will occur in the jobs chosen by the \emph{algorithm} (and thus it will learn their distortion well), while the overestimations will occur in the jobs chosen by the \emph{optimal solution} (and thus the algorithm will not learn of their overestimation).

\paragraph{The $\zza$ algorithm.}
This failure of $\thres$ in the case of overestimations seems inherent to any algorithm that attempts to learn the distortion from the input.
The $\zza$ algorithm thus takes a different approach, and ignores observed distortion completely.

Our $\zza$ algorithm maintains a set of partial jobs, i.e. jobs that have been processed to some degree.
Each such job is either a \zig{} job or a \zag{} job.
The algorithm always works on the partial job of the lowest class; however, under some circumstances this partial job could decide to ``appoint'' a full job of a lower class to be partial.
\zig{} jobs always appoint \zag{} jobs, and do so according to $\sept$. \zag{} jobs always appoint \zig{} jobs, and do so according to $\sr$.

However, for the algorithm to be competitive, we require another component -- the \zigzag{} jobs.
When a full job is released between a \zag{} job and the \zig{} job that appointed it, the \zag{} job turns into a \zigzag{} job.
\zigzag{} jobs still appoint \zig{} jobs (like they did as \zag{} jobs), but do so according to $\sept$ rather than $\sr$.

\subsection{Related Work}

The observation that exact processing times are usually not available in practice also motivated the \emph{nonclairvoyant model}.
In this model, the algorithm is not given the processing times of jobs, and must process them blindly; the algorithm only becomes aware of the processing time of a job when the job is completed.
In a sense, the nonclairvoyant model goes to the other extreme -- it assumes that \emph{nothing} is known about the processing time of a job.
This strictness of the model comes at a cost -- the best algorithm for nonclairvoyant flow-time scheduling in~\cite{Becchetti2001} requires randomization, and has a competitive ratio that grows with the number of jobs $n$ (specifically, $O(\log n)$).
In~\cite{Motwani1994}, a matching lower bound of $\Omega(\log n)$-competitiveness was shown for randomized algorithms, as well as a lower bound of $\Omega(n^{1/3})$-competitiveness for deterministic algorithms.
Other considerations of the nonclairvoyant scheduling model appear in~\cite{KalyanasundaramP1995,Becchetti2001,DBLP:journals/ipl/KimC03a,DBLP:journals/algorithmica/BansalDKS04,Im2014,Im2017}.

A similar model to the model considered in this paper is the semiclairvoyant model in~\cite{DBLP:journals/tcs/BecchettiLMP04}.
In this model, one knows the power-of-2 class of the processing time, but not the exact processing time.
This model is different from the robust model since the distortion is constant and conforms to the boundaries of classes.

A generalization of the total flow time goal is total \emph{weighted} flow time, in which each job also has a weight which scales its flow time.
This goal function was studied in e.g.~\cite{DBLP:conf/stoc/ChekuriKZ01,DBLP:journals/talg/BansalD07,DBLP:conf/focs/AzarT18,BECCHETTI2006339}.
The best-known upper bounds are logarithmic in the parameters of the input (processing-time ratio, weight ratio, density ratio).
Bansal and Chan~\cite{DBLP:conf/soda/BansalC09} showed that dependence on these parameters is necessary.

The robust model for scheduling was introduced in~\cite{DBLP:journals/corr/abs-2103-05604}.
For minimizing total flow time, the paper introduced $\dstr$-robust, $O(\dstr^2)$ competitive algorithms for every $\dstr$.
In addition, the paper also introduced robust algorithms for weighted flow time with polynomial dependence (quadratic and cubic) on the distortion and logarithmic dependence on the parameters.
All of these robust algorithm require a priori knowledge of the distortion.

A related emerging field is algorithms with predictions, in which algorithms are augmented with predictions of varying accuracy which relate to the incoming  input.
Such algorithms are expected to produce a competitive ratio which is a function of the accuracy of the predictions.
In particular, scheduling with predictions for minimizing total completion time was considered in~\cite{purohit2018improving,10.1145/3409964.3461790} (these papers assume that all jobs are given at time $0$).
Scheduling with predictions has also been studied in the speed-scaling model~\cite{bamas2020learning}, under stochastic arrival assumptions~\cite{mitzenmacher:LIPIcs:2020:11699} and for load balancing~\cite{LattanziLMV20,Lavastida2020LearnableAI}.
Additional work on algorithms with predictions can be found in~\cite{mitzenmacher_vassilvitskii_2021}.

\paragraph{Paper Organization.}
The $\zza$ algorithm is presented and analyzed in \cref{sec:ZZA}.
The $\thres$ algorithm is presented and analyzed in \cref{sec:DLA}.
The lower bound for inputs with distortion $\dstr$ is given in \cref{sec:LB}.
Bad cases for various existing algorithms are shown in \cref{sec:BadCases}.

    \section{Preliminaries}
    \label{sec:Prelim}
    
In the problem we consider, jobs arrive over time.
Each job $\req$ must be processed for  $\pt{\req}$ time units until its completion ($\pt{\req}$ is called the \emph{processing time} or \emph{volume} of $\req$).
We denote the release time of $\req$ by $\rlt{\req}$.

Upon the release of a job $q$, the algorithm is \emph{not} given the processing time $\pt{q}$; instead, the algorithm is given an estimate $\ept{q}$ to the processing time of $q$.

We define $\udstr := \max_{\text{job }q} \frac{\pt{q}}{\ept{q}}$, the maximum underestimation factor of a job in the input.
We also define $\ddstr := \max_{\text{job }q} \frac{\ept{q}}{\pt{q}}$, the maximum overestimation factor of a job in the input.
(we demand that $\udstr,\ddstr \ge 1$; if this is not the case for one of these factors, define that factor to be $1$.)
It thus holds for every job $q$ that $\frac{\ept{q}}{\ddstr} \le \pt{q} \le \udstr\cdot \ept{q}$.
Finally, we define $\dstr := \udstr\cdot \ddstr$, the distortion parameter of the input.
Note that these parameters $\udstr,\ddstr,\dstr$ are functions of the online input, and thus the algorithm has no prior knowledge of them.

The goal of an algorithm $\alg$ is to minimize the \emph{total flow time}, which is the sum over jobs of their time in the system, or
$\sum_{\text{job }q} (C^{\alg}_q - \rlt{q})$ where $C^{\alg}_q$ is the completion time of $q$ in the algorithm.
Denoting by $\sgen(t)$ the number of pending jobs in the algorithm at $t$ (i.e. jobs that were released but not yet completed), an equivalent definition of flow time is $\int_0^\infty \sgen(t) dt$.

The following definition of job classes is used throughout the paper.
\begin{defn}[job class]
    We define the \emph{class} of a job $q$, denoted $\cls{q}$, to be the unique integer $i$ such that $\ept{q} \in \IR{2^i}{2^{i+1}}$.
\end{defn}
Note that this definition refers to the \emph{estimated} processing times provided in the input (rather than actual processing times, or remaining processing times).
In particular, the class of a job does not change over time.


    \section{The \texorpdfstring{$\zza$}{ZigZag} Algorithm}
    \label{sec:ZZA}
    
In this section, we describe and analyze the $\zza$ algorithm, a distortion-oblivious algorithm which is $O(\dstr\log\dstr)$-competitive for every input with distortion at most $\dstr$, for every $\dstr >1$.

\subsection{The Algorithm}

We now describe the following algorithm for the robust scheduling problem.
The algorithm marks some set of the pending jobs as \emph{partial} jobs.
These partial jobs are the only jobs that may undergo processing in the algorithm; the remaining jobs, called \emph{full} jobs, have not undergone any processing.

The algorithm always works on the minimum-class partial job $q$ (there can be at most one partial job in any class).
If there exist full jobs of lower classes than $q$, the job $q$ might choose to ``appoint'' the minimum-class full job to be partial (this newly-appointed job is now the minimum-class partial job, and will thus be processed next).

This decision depends on the type of $q$ as a partial job.
Each partial job is one of the three types \zig{}, \zag{} and \zigzag{}.
\begin{enumerate}
    \item If $q$ is a \zig{} job, it immediately appoints any smaller-class job, and marks this job as \zag{}.
    \item If $q$ is a \zag{} job, it only appoints the minimum-class full job when there exist (at least) two jobs of class less than $q$.
    This appointed job is then marked as $\zig{}$.
    \item If $q$ is a \zigzag{} job, it immediately appoints any smaller-class job (like \zig{} jobs do), but marks this job as \zig{} (like \zag{} jobs do).
\end{enumerate}
\zig{} and \zag{} jobs are such immediately from the point of their appointment.
\zigzag{} jobs are created in the following way: when a \zag{} job $q$ is the minimal-class partial job, and there exists a full job between the class of $q$ and the next-higher-class partial job, the \zag{} job $q$ morphs into a \zigzag{} job.

\begin{algorithm}[h]
    \caption{\label{alg:ZZA} $\zza$ Algorithm}
    \While{there exist pending jobs}{

        \If{there is no partial job in the algorithm}{
            Mark an arbitrary minimum-class pending job as a partial \zig{} job.


            \Continue to the next iteration of the loop.
        }

        \BlankLine

        Let $q$ be the minimum-class partial job in the algorithm.

        Let $q'$ be the minimum-class job in the algorithm (partial or full).

        \BlankLine

        \If{$q$ is a \zig{} job}{
            \If{$\cls{q'} < \cls{q}$}{
                Mark $q'$ as a partial \zag{} job.

                \Continue to the next iteration of the loop.
            }
%
%
%
        }

        \BlankLine

        \ElseIf{$q$ is a \zag{} job}{
            Let $\hat{q}$ be the partial job of the smallest class such that $\cls{\hat{q}}>\cls{q}$.

            \If{there exists a full job of class in $[\cls{q},\cls{\hat{q}}]$}{
                Change $q$ from a \zag{} job to a \zigzag{} job.

                \Continue to the next iteration of the loop.
            }

            \If{$\cls{q'} < \cls{q}$ \And there exists some other job $q''\neq q'$ such that $\cls{q''} < \cls{q}$}{
%

                Mark $q'$ as a partial \zig{} job. \label{line:ZagMakesZig}

                \Continue to the next iteration of the loop.
            }
        }

        \BlankLine

        \Else(\tcp*[h]{$q$ is a \zigzag{} job}){
            \If{$\cls{q'} < \cls{q}$}{
                Mark $q'$ as a partial \zig{} job.

                \Continue to the next iteration of the loop.
            }
%
%
%
        }

        \BlankLine

        Process $q$.
    }
\end{algorithm}

A visualization of the $\zza$ algorithm can be found in \cref{fig:ZZA_ZigZagState}.
In this visualization, the pending jobs are shown inside their classes (the jobs appear as geometric shapes).
The jobs can be either full jobs (hollow circles), $\zig$ jobs (blue squares), $\zag$ jobs (red rhombus), or $\zigzag$ jobs (purple star).
The gray arrows show appointment, where each partial job has necessarily been appointed by the next-higher-class partial job.
Note that there is always at most one partial job per class, and note the alternation between $\zig$ (blue) jobs and $\zag$ or $\zigzag$ jobs (red and purple).
In addition, note that there is always a full job between a $\zig$ job and a $\zigzag$ job appointed by it.
Similarly, there is always a full job between a $\zag$ job and a $\zig$ job appointed by it.

An additional visualization of the evolution of job types is given in \cref{fig:ZZA_JobEvolution}.
\begin{figure}
    \includegraphics[width=\columnwidth]{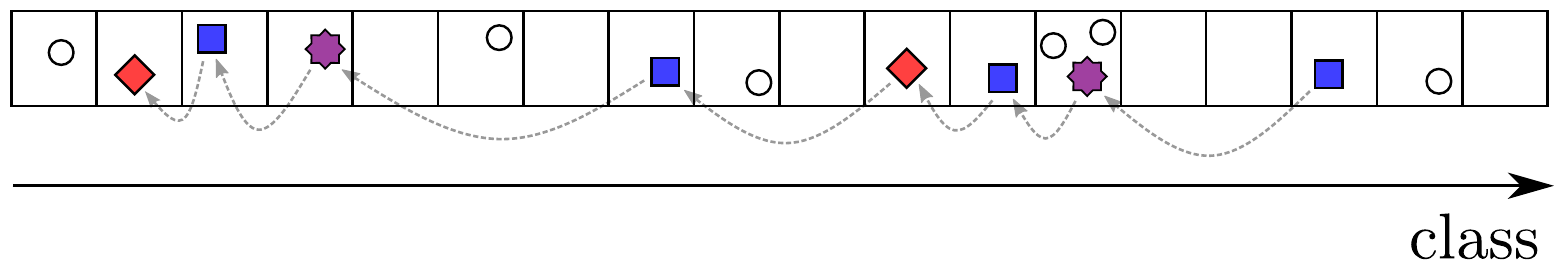}
    \caption{\label{fig:ZZA_ZigZagState}A Possible State of the $\zza$ Algorithm (square=\zig{}, rhombus=\zag{}, star=\zigzag{}, circle=full)}
\end{figure}
\begin{figure*}
    \begin{center}
        \hspace*{\fill}
        \subfloat[][\label{subfig:ZZA_JobEvolution1}$q_1$ is released and is immediately appointed \zig{}.]{
            \includegraphics[width=0.4\columnwidth]{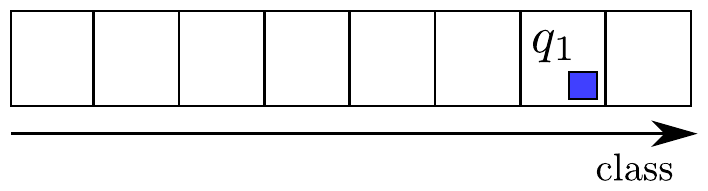}
        }
        \hspace*{\fill}
        \subfloat[][\label{subfig:ZZA_JobEvolution2}$q_2$ is released, and is immediately appointed \zag{} by $q_1$.]{
            \includegraphics[width=0.4\columnwidth]{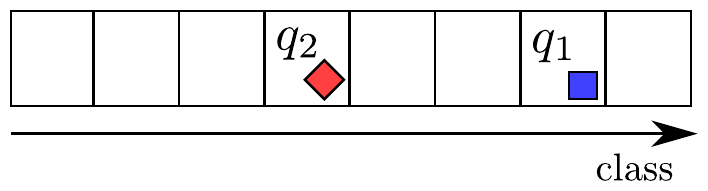}
        }
        \hspace*{\fill}

        \hspace*{\fill}
        \subfloat[][\label{subfig:ZZA_JobEvolution3}$q_3$ is released. $q_2$ does not appoint $q_3$, since there is only a single job of class less than $\cls{q_2}$.]{
            \includegraphics[width=0.4\columnwidth]{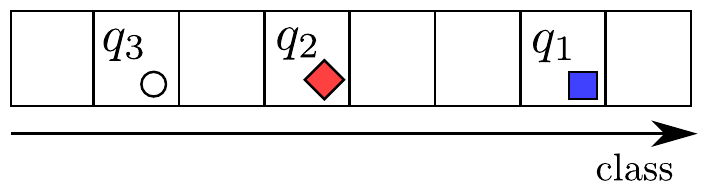}
        }
        \hspace*{\fill}
        \subfloat[][\label{subfig:ZZA_JobEvolution4}$q_4$ is released, which immediately turns $q_2$ from \zag{} to \zigzag{}. Then, $q_2$ immediately appoints $q_3$ as $\zig$. ]{
            \includegraphics[width=0.4\columnwidth]{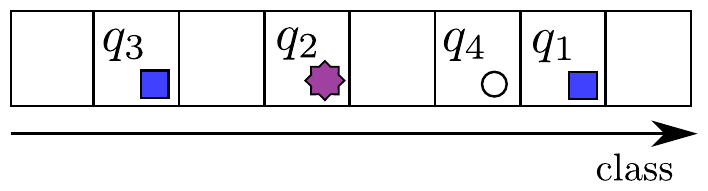}
        }
        \hspace*{\fill}
    \end{center}
    \caption{\label{fig:ZZA_JobEvolution} A Possible Evolution of Jobs in \cref{alg:ZZA}}
\end{figure*}

In the remainder of this section, we prove the following theorem regarding \cref{alg:ZZA}.

\begin{thm}
    \label{thm:ZZA_Competitiveness}
    For every $\dstr$, \cref{alg:ZZA} is $O(\dstr \log \dstr)$-competitive for inputs with distortion $\dstr$.
\end{thm}

\subsection{Analysis}

We prove \cref{thm:ZZA_Competitiveness} through the following lemma, which states a property known as local competitiveness.

\begin{lem}
    \label{lem:ZZA_LocalCompetitiveness}
    At any time $t$, it holds that $\sgen(t) \le O(\dstr\log\dstr)\cdot \sgen*(t)$.
\end{lem}

If \cref{lem:ZZA_LocalCompetitiveness} holds, \cref{thm:ZZA_Competitiveness} follows simply through integrating over $t$.
The remainder of this section focuses on proving \cref{lem:ZZA_LocalCompetitiveness}; we henceforth fix time $t$ for the remainder of this analysis.
In addition, note that if $\sgen*(t) =0$ then since the algorithm is non-idling, it must also be that $\sgen(t) =0$, and \cref{lem:ZZA_LocalCompetitiveness} holds.
Thus, we henceforth assume that $\sgen*(t) \ge 1$.

\subsubsection{Bounding \texorpdfstring{$\sgen(t)$}{delta(t)} by Far-Behind Classes}

In this subsection, we bound the number of living jobs in the algorithm by the number of ``far-behind'' classes, which we define soon.

\begin{defn}[volume and living jobs notation]
    \label{defn:ZZA_VolumeAndDeltaDefs}
    For every class $i$ and time $\tau$, we define:
    \begin{itemize}
        \item $\Vgen[=i](\tau)$ ($\Vgen[\le i](\tau)$) to be the total remaining volume at $\tau$ of jobs of class exactly $i$ (at most $i$) in the algorithm.
        \item Similarly, we define $\Vgen*[=i](\tau)$ ($\Vgen*[\le i](\tau)$) to be the total remaining volume at $\tau$ of jobs of class exactly $i$ (at most $i$) in the optimal solution.
        \item Finally, we define $\Delta \Vgen[=i](\tau) :=\Vgen[=i](\tau) - \Vgen*[=i](\tau)$ (and similarly $\Delta \Vgen[\le i](\tau) := \Vgen[\le i](\tau) - \Vgen*[\le i](\tau)$).
    \end{itemize}
    We use similar subscript notation with $\sgen$ to refer to the number of pending jobs of some class (or class range), and use superscript $*$ to refer to that amount in the optimal solution.
\end{defn}

\begin{defn}[far-behind classes]
    \label{defn:ZZA_FB}
    For every class $i$, we say that $i$ is \emph{far behind} at $t$ if $\Delta\Vgen[\le i](t) \ge \frac{2^i}{\ddstr}$.

    We also denote by $S$ the set of far-behind classes at $t$.
\end{defn}

In this subsection, we prove the following lemma.

\begin{lem}
    \label{lem:ZZA_NumJobsBoundByFB}
    $\sgen(t)  \le O(\dstr) \cdot \sgen*(t) + O(\dstr) \cdot \ps{S}$
\end{lem}

We first make some observations regarding the algorithm.
\begin{obs}
    \label{obs:ZZA_OnlyOnePartialPerClass}
    At every time $\tau$, and for every class $j$, there exists at most one partial job of class $j$ at $\tau$.
\end{obs}
\begin{obs}
    \label{obs:ZZA_ZigZagAlternation}
    At any time $\tau$, denote by $q_1, q_2, \cdots,q_k$ the partial jobs in the algorithm, by order of decreasing class ($q_1$ has the largest class).
    Then each job in $\pc{q_1,q_3,q_5,\cdots }$ is a \zig{} job, and each job in $\pc{q_2, q_4, q_6, \cdots}$ is either a \zag{} job or a \zigzag{} job.
\end{obs}

\begin{obs}
    \label{obs:ZZA_PartialIsBarrier}
    Since a job $q$ becomes partial and until its completion, the algorithm only processes job $q$ or jobs of class less than $\cls{q}$.
    Moreover, during that time interval no full job of class at least $\cls{q}$ becomes partial.
\end{obs}

\begin{obs}
    \label{obs:ZZA_OnlyOneFull}
    Suppose a partial job $q$ is being processed at time $\tau$.
    Then there is at most one job of class less than $\cls{q}$ (which, if exists, is necessarily full).
\end{obs}

\begin{defn}
    \label{defn:ZZA_JobNumDefs}
    For every time $\tau$, we denote:
    \begin{itemize}
        \item the number of full jobs at $t$ by $\sgen<f>(\tau)$.
        \item the number of partial jobs at $t$ by $\sgen<p>(\tau)$.
    \end{itemize}
\end{defn}

\begin{prop}
    \label{prop:ZZA_FullBetweenFourPartial}
    Consider any four partial jobs $q_1,q_2,q_3,q_4$ in the algorithm at $t$ such that:
    \begin{itemize}
        \item $\cls{q_1} < \cls{q_2} < \cls{q_3} < \cls{q_4}$
        \item $q_1,q_2,q_3,q_4$ are consecutive: that is, for every $i$ there exists no partial job $q'$ such that $\cls{q_i} < \cls{q'} < \cls{q_{i+1}}$.
    \end{itemize}
    Then there exists a full job $q$ such that $\cls{q_1} \le \cls{q} \le \cls{q_4}$.
\end{prop}
\begin{proof}
    Using \cref{obs:ZZA_ZigZagAlternation}, either $q_1,q_3$ are \zig{} jobs or $q_2,q_4$ are \zig{} jobs.

    Assume henceforth that $q_1,q_3$ are \zig{} jobs.
    This implies that $q_2$ is either a \zag{} job or a \zigzag{} job.

    Note that among $q_1,q_2,q_3$, $q_3$ became partial first, then $q_2$ and then $q_1$; any other order would contradict \cref{obs:ZZA_PartialIsBarrier}.

    If $q_2$ is a \zag{} job, then consider the point in time $t'$ in which $q_1$ became partial.
    At that point, there was no pending job of class lower than $\cls{q_1}$.
    Moreover, there was no partial job of class lower than $\cls{q_2}$ at $t'$ (otherwise, that job would be alive in $t$, contradicting the fact that $q_1, q_2$ are consecutive).
    This implies that $q_1$ became partial when $q_2$ was the minimum-class partial job, as part of \cref{line:ZagMakesZig} in the algorithm.
    \cref{line:ZagMakesZig} ran at $t'$ because there had been at least two full jobs of class smaller than $\cls{q_2}$.
    Thus, after $q_1$ became partial, there remained a full job $q$ in the class range $\IR{\cls{q_1}}{\cls{q_2}}$.
    Applying \cref{obs:ZZA_PartialIsBarrier} to $q_1$, we have that $q$ is still full at $t$.

    Otherwise, $q_2$ is a \zigzag{} job.
    In this case, let $t'$ be the point in time in which $q_2$ became a \zigzag{} job.
    This happened since there was a full job $q$ in the class range $[\cls{q_2},\cls{q_3}]$; applying \cref{obs:ZZA_PartialIsBarrier} to $q_2$ implies that $q$ is still full at $t$.

    Overall, we have that if $q_1,q_3$ are \zig{} jobs, then there exists a full job in the class range $[\cls{q_1},\cls{q_3}]$.
    In the second case, in which $q_2,q_4$ are \zig{} jobs, the same argument yields that there exists a full job in the class range $[\cls{q_2},\cls{q_4}]$.
    Thus, in both cases the proposition holds.
\end{proof}
The following corollary is immediate from \cref{prop:ZZA_FullBetweenFourPartial} by partitioning the partial jobs into groups of four.
\begin{cor}
    \label{cor:ZZA_PartialBoundedByFull}
    It holds that $\sgen<p>(t) \le 4\sgen<f>(t) + 3$.
\end{cor}

\Cref{cor:ZZA_PartialBoundedByFull} allows us to focus on bounding the total number of full jobs, which we now proceed to do.

\begin{lem}
    \label{lem:ZZA_VolumeConversion}
    It holds that
    \[
        \sgen<f>(t)  \le O(\dstr) \cdot \sgen*(t) + O(\dstr) \cdot \ps{S}
    \]
\end{lem}
\begin{proof}
    We omit $t$ from $\Vgen$-notation and $\sgen$-notation in the following.
    Let $i_{\min}, i_{\max}$ be the minimum and maximum class of a job in the input, respectively.
    \begin{align}
        \label{eq:ZZA_FullJobBound1}
        \sgen<f> &= \sum_{i=i_{\min}}^{i_{\max}}\sgen[=i]<f> \\
        &\le \sum_{i=i_{\min}}^{i_{\max}} \floor{\frac{\Vgen[=i]}{\frac{2^i}{\ddstr}}} \nonumber\\
        &\le \sum_{i=i_{\min}}^{i_{\max}} \floor{\frac{\ddstr\Vgen*[=i] + \ddstr\cdot \Delta \Vgen[=i]}{2^i}} \nonumber\\
        &\le \sum_{i=i_{\min}}^{i_{\max}} \ceil{\frac{\ddstr\Vgen*[=i]}{2^i}} +\sum_{i=i_{\min}}^{i_{\max}} \floor{\frac{\ddstr\cdot \Delta \Vgen[=i]}{2^i}}\nonumber
    \end{align}
    where the first inequality is due to the fact that the minimum volume of a full job of class $i$ is $\frac{2^i}{\ddstr}$ and the third inequality is due to the simple arithmetic trait that $\floor{x+y} \le \ceil{x} + \floor{y}$.

    Observe the expression $\sum_{i=i_{\min}}^{i_{\max}} \ceil{\frac{\ddstr\Vgen*[=i]}{2^i}}$.
    If $\sgen*[=i] \ge 1$, we have
    \[
        \ceil{\frac{\ddstr\Vgen*[=i]}{2^i}} \le \frac{\ddstr\Vgen*[=i]}{2^i} + 1 \le \frac{\ddstr}{2^i}\cdot 2^{i+1}\udstr\sgen*[=i] + 1 = 2\dstr\sgen*[=i] + 1 \le (2\dstr + 1)\sgen*[=i]
    \]
    Otherwise, $\sgen*[=i] = 0$, in which case $\ceil{\frac{\ddstr\Vgen*[=i]}{2^i}} = 0 = (2\dstr+1)\sgen*[=i]$.

    Thus, $\sum_{i=i_{\min}}^{i_{\max}} \ceil{\frac{\ddstr\Vgen*[=i]}{2^i}} \le (2\dstr+1)\sgen*$.
    Plugging into \cref{eq:ZZA_FullJobBound1}, we have
    \begin{align}
        \label{eq:ZZA_FullJobBound2}
        \sgen<f> &\le (2\dstr+ 1)\sgen*  + \sum_{i= i_{\min}}^{i_{\max}} \floor{\frac{\ddstr \cdot \Delta \Vgen[\le i] - \ddstr \cdot \Delta \Vgen[\le i-1]}{2^i}} \\
        &\le (2\dstr+ 1)\sgen* + \sum_{i= i_{\min}}^{i_{\max}} \pr{\floor{ \frac{\ddstr\Delta \Vgen[\le i]}{2^i}} -\floor{ \frac{\ddstr\Delta \Vgen[\le i-1]}{2^i}}} \nonumber\\
        &= (2\dstr+ 1)\sgen* + \sum_{i= i_{\min}}^{i_{\max}-1} \pr{\floor{\frac{\ddstr\Delta \Vgen[\le i]}{2^i}} -\floor{ \frac{\ddstr\Delta \Vgen[\le i]}{2^{i+1}}}} \nonumber
    \end{align}
    where the second inequality is due to the arithmetic trait that $\floor{x-y} \le \floor{x} - \floor{y}$, and the equality is through rearranging a telescopic sum, while observing that $\Delta\Vgen[\le i_{\min}-1] = \Delta \Vgen[\le i_{\max}] = 0$ (since the algorithm is non-idling).

    Now, we claim that for every $i$ it holds that $\Delta \Vgen[\le i] \le 2^{i+1}\udstr$.
    To see this, consider the last time before $t$ in which we worked on a job of class strictly more than $i$, and denote this time by $\lastt{i}$.
    At $\lastt{i}$, there was at most one pending job of class at most $i$ (\cref{obs:ZZA_OnlyOneFull}), and that job thus had at most $2^{i+1}\udstr$ volume.
    Since from $\lastt{i}$ the algorithm only worked on jobs of class at most $i$, $\Delta \Vgen[\le i](t) \le \Delta \Vgen[\le i](\lastt{i}) \le \Vgen[\le i](\lastt{i}) \le 2^{i+1}\udstr$.

    Thus, observe that the expression $\floor{\frac{\ddstr\Delta \Vgen[\le i]}{2^i}} -\floor{ \frac{\ddstr\Delta \Vgen[\le i]}{2^{i+1}}}$ is:
    \begin{itemize}
        \item maximized when $\Delta \Vgen[\le i] = 2^{i+1}\udstr$, and its value then is at most $\dstr + 1$.
        \item only positive when $\Delta \Vgen [\le i] \ge \frac{2^i}{\ddstr}$, i.e. when $i$ is far behind.
    \end{itemize}
    From these two observations, we can bound $\floor{\frac{\ddstr\Delta \Vgen[\le i]}{2^i}} -\floor{ \frac{\ddstr\Delta \Vgen[\le i]}{2^{i+1}}}$ by $(\dstr+1) \cdot \mathbb{I}(i\text{ is far behind})$.
    Plugging into \cref{eq:ZZA_FullJobBound2}, we get that the weight of the full jobs $\sgen<f>$ is at most
    \[
        O(\dstr) \cdot \sgen* + O(\dstr)  \cdot \ps{S}
    \]
    as required.
\end{proof}

We can now prove \cref{lem:ZZA_NumJobsBoundByFB}.

\begin{proof}[Proof of \cref{lem:ZZA_NumJobsBoundByFB}]
    The lemma results immediately from \cref{cor:ZZA_PartialBoundedByFull,lem:ZZA_VolumeConversion}.
\end{proof}

\subsubsection{Bounding Far-Behind Classes}

To complete the bounding of the algorithm's cost, it is thus enough to bound the size of the set $S$, i.e. number of far-behind classes.

\begin{lem}
    \label{lem:ZZA_FBBoundedByOPT}
    $\ps{S} \le O(\log \dstr)\cdot \sgen*(t)$.
\end{lem}

\begin{defn}
    \label{defn:ZZA_Seperator}
    Define $\sep := \ceil{\log \dstr} + 1$.
\end{defn}
Intuitively, $\sep$ is the minimum number such that a job $q$ of class $i$ necessarily has less processing time than a job $q'$ of class $i+\sep$:
\[
    \pt{q'} \ge \frac{2^{i+\sep}}{\ddstr} \ge \frac{2^{i+\log \dstr+1}}{\ddstr} = \dstr\cdot \frac{2^{i+1}}{\ddstr} = \udstr\cdot 2^{i+1} > \pt{q}
\]

We now perform a sparsification of $S$ to obtain the set $S'$ in the following manner.
\begin{enumerate}
    \item Add the minimum class in $S$ to $S'$.
    \item While possible, add the minimum class that is greater than the last-added class by at least $2\sep$.
\end{enumerate}

\begin{obs}
    \label{obs:ZZA_SparsificationProperties}
    Observe that:
    \begin{enumerate}
        \item $\ps{S} \le 2\sep \cdot \ps{S'}$.
        \item $\ps{i_1 -i_2} \ge 2\sep$ for every $i_1,i_2 \in S'$.
    \end{enumerate}
\end{obs}

A visualization of this sparsification process is given in \cref{fig:ZZA_FBSparsification}.
\Cref{subfig:ZZA_FBClasses} shows the far-behind classes of $S$ in red.
\Cref{subfig:ZZA_FBClassesSparse} shows the classes of the sparsified set $S'$ in red, where the classes of $S\backslash S'$ are faded (in this example $\sep=2$).
The purple lines show the range in which classes are excluded.
Note that the distance between any two classes in $S'$ is at least $2\sep$ (which equals $4$ in this figure).

\begin{figure*}
    \begin{center}
        \subfloat[][\label{subfig:ZZA_FBClasses}The Far-Behind Classes in $S$.]{
            \includegraphics[width=\columnwidth]{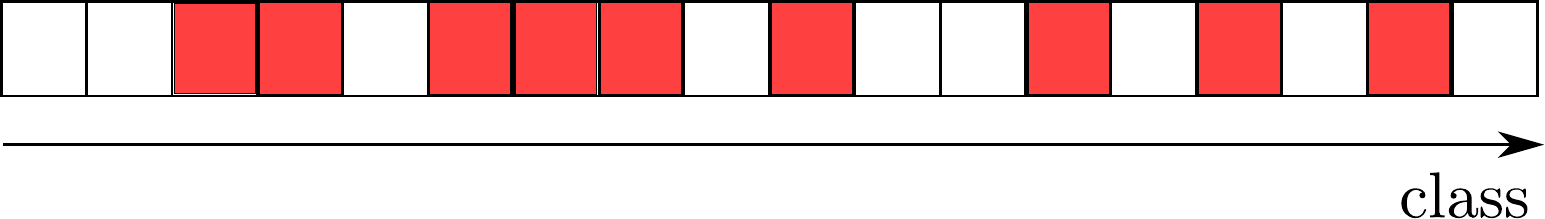}
        }
        \vspace{4em}
        \subfloat[][\label{subfig:ZZA_FBClassesSparse}The Sparsified Set $S'$.]{
            \includegraphics[width=\columnwidth]{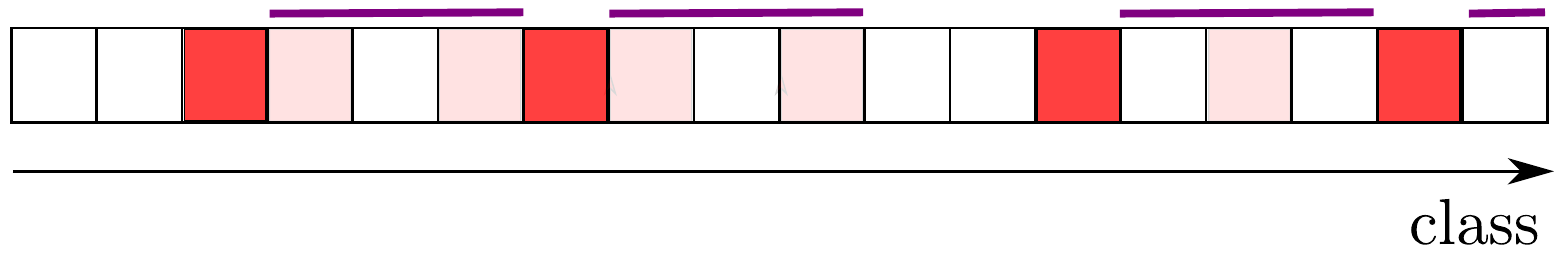}
        }
    \end{center}
    \caption{\label{fig:ZZA_FBSparsification} The Far-Behind Sparsification Process}
\end{figure*}

We henceforth focus on bounding $\ps{S'}$.

\begin{defn}
    \label{defn:ZZA_LastProcTime}
    For every class $i$, we denote by $\lastt{i}$ the last time until $t$ in which the algorithm worked on a job of class strictly more than $i$ ($\lastt{i}:=0$ if this never happened).
\end{defn}

\begin{prop}
    \label{prop:ZZA_FarBehindHasJob}
    If class $i$ is far behind at $t$, then there exists a pending job at $t$ of some class in $\I{i-\sep}{i}$.
\end{prop}
\begin{proof}
    Consider time $\lastt{i}$, the last time before $t$ that the algorithm worked on a job of class strictly more than $i$.

    Note that during $\I{\lastt{i}}{t}$ the algorithm only worked on jobs of class at most $i$, and $i$ is far behind at $t$; thus, the total volume of jobs of class at most $i$ at $\lastt{i}$ must be at least $\frac{2^i}{\ddstr}$.

    Now note that at $\lastt{i}$ the algorithm worked on the minimum-class partial job $q$ which is of class more than $i$.
    If this job $q$ is \zig{} or \zigzag{} at $\lastt{i}$, this is a contradiction, since it would imply that there are no jobs of class at most $i$ at $\lastt{i}$.
    Thus, the job $q$ is a \zag{} job -- but this implies that there is only a single job $q'$ of class at most $i$ at $\lastt{i}$.
    This job must thus have at least $\frac{2^i}{\ddstr}$ volume, which would imply that its class is greater than $i-\sep$, and thus in $\I{i-\sep}{i}$.

    Now, suppose for contradiction that no job of class in $\I{i-\sep}{i}$ exists at $t$.
    This implies that at some time $\tau$ in $\I{\lastt{i}}{t}$, the last job of this class range was completed.
    At that time, there was at most one other job of class $\le i$, and that job (if it exists) was of class $\le i-\sep$, and thus had strictly less than $\frac{2^i}{\ddstr}$ volume.
    Since $\tau > \lastt{i}$, we have
    \[
        \Delta \Vgen[\le i](t) \le \Delta \Vgen[\le i](\tau) <\frac{2^i}{\ddstr}
    \]
    which is a contradiction to $i$ being far behind at $t$.
\end{proof}

\begin{prop}
    \label{prop:ZigTemporalOnion}
    If a \zig{} job $q$ of some class $i$ is pending at $t$, it holds that $\Delta\Vgen[\le i-1](t) \le 0$.
\end{prop}
\begin{proof}
    \Cref{obs:ZZA_PartialIsBarrier} implies that from the time $q$ became partial, the algorithm only worked on $q$ or on jobs of class at most $i-1$.
    Consider the last time $t'$ in which the algorithm worked on $q$: at that time, there were no jobs of classes at most $i-1$, and from that time the algorithm only worked on jobs of class at most $i-1$.
    Thus,
    \[
        \Delta \Vgen[\le i-1](t) \le \Delta \Vgen[\le i-1](t') \le 0
    \]
\end{proof}
\begin{prop}
    \label{prop:ZZA_OPTBetweenFullAndFB}
    Let $i, i'$ be two classes such that $i' \le i - \sep$, $i'$ is far-behind, and there exists a full job of class $i$ in the algorithm at $t$.
    Then the optimal solution has a job alive in the class range $\I{i'}{i}$.
\end{prop}
\begin{proof}
    Denote by $q$ the full job of class $i$ in the algorithm at $t$.
    If $q$ is pending in the optimal solution at $t$, we are done; henceforth assume that the optimal solution has completed $q$ by time $t$.

    We now aim to find a class $j \in [i-\sep,i]$ such that $\Delta\Vgen[\le j](t) \le 0$, and claim that this would complete the proof.
    To prove this claim, assume that there exists such a $j$.
    Observe that $\Delta\Vgen[\le i'](t) >0$ (since $i'$ is far behind); thus, it cannot be that $i'=j=i-\sep$.
    Therefore, it holds that $i' < j$, which implies
    $\Delta \Vgen[\in \I{i'}{j}](t) = \Delta\Vgen[\le j](t) - \Delta \Vgen[\le i'](t) < 0$, where the subscript $\in\I{i'}{j}$ restricts the volume to jobs of classes in $\I{i'}{j}$.
    But this implies that the optimal solution must have a pending job at $t$ of some class in $\I{i'}{j}$, and thus in $\I{i'}{i}$.
    Hence, the claim holds.

    We now continue in proving the proposition.
    We consider $\lastt{i-\sep}$ relative to $\rlt{q}$, and observe the following two cases.

    \textbf{Case 1: $\rlt{q} \ge \lastt{i-\sep}$.} In this case, note that at $\lastt{i-\sep}$ there existed at most one job of class at most $i-\sep$ due to \cref{obs:ZZA_OnlyOneFull}.
    Such a job had volume less than $2^{i-\sep+1}\udstr$ which is at most $\frac{2^i}{\ddstr}$.
    Thus, $\Vgen[\le i-\sep](\lastt{i-\sep}) \le \frac{2^i}{\ddstr}$.

    In the time interval $\I{\lastt{i-\sep}}{t}$, the algorithm only worked on jobs of class at most $i-\sep$, while the optimal solution started and completed $q$.
    This implies that
    \[
        \Delta \Vgen[\le i-\sep](t) \le \Delta \Vgen[\le i-\sep](\lastt{i-\sep}) - \pt{q} \le 0
    \]
    which completes the proof for this case.

    \textbf{Case 2: $\rlt{q} < \lastt{i-\sep}$}.
    In this case, consider the job $q'$ being processed at $\lastt{i-\sep}$.
    From the definition of $\lastt{i-\sep}$, it must be that $\cls{q'} > i-\sep$.
    If there exists no job of class $\le i-\sep$ at $\lastt{i-\sep}$, we are done, since the algorithm only works on such jobs in $\I{\lastt{i-\sep}}{t}$, and thus
    \[
        \Delta \Vgen[\le i-\sep](t) \le \Delta \Vgen[\le i-\sep](\lastt{i-\sep}) \le 0
    \]
    Otherwise, there exists a job $r$ of class at most $i-\sep$ at $\lastt{i-\sep}$, which implies that $q'$ is a \zag{} job at $\lastt{i-\sep}$.
    Let $q''$ be the consecutive partial job to $q'$, and note that $q''$ is a \zig{} job (observe that a \zag{} job always has a consecutive job).
    Now note that:
    \begin{itemize}
        \item $\cls{q'}<i$; otherwise, the existence of both $r$ and $q$ would prevent $q'$ from being processed at $\lastt{i-\sep}$.
        \item $\cls{q''} < i$; otherwise, $q'$ (which is of class less than $i$) would have become a \zigzag{} job by seeing $q$, in contradiction to being \zag{} at $\lastt{i-\sep}$.
    \end{itemize}
    Since the algorithm does not work on a job of class more than $i-\sep$ after $\lastt{i-\sep}$, the \zig{} job $q''$ remains pending at $t$.
    Using \cref{prop:ZigTemporalOnion} implies that $\Delta \Vgen[\le \cls{q''}-1](t) \le 0$.
    Since $\cls{q''}-1 \in [i-\sep, i-1]$, we are done.
\end{proof}

\begin{lem}
    \label{lem:ZZA_OptJobInFiveFB}
    Let $i_1,i_2,\cdots,i_5 \in S'$ be five classes such that $i_1<i_2<\cdots<i_5$.
    Then there exists a pending job in the optimal solution at $t$ of class in the range $[i_1,i_5]$.
\end{lem}
\begin{figure}
    \begin{center}
        \includegraphics{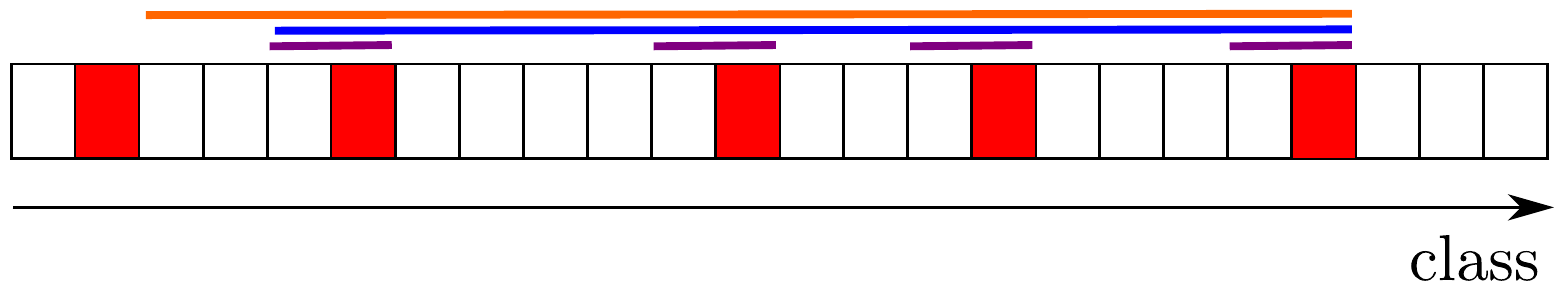}
    \end{center}
    This figure shows the proof of \cref{lem:ZZA_OptJobInFiveFB}.
    The figure shows five far-behind classes in $S'$.
    \Cref{prop:ZZA_FarBehindHasJob} implies that there exists a pending job in each of the four purple-colored class segments.
    \Cref{prop:ZZA_FullBetweenFourPartial} implies that there is a full job in the blue-colored segment.
    Finally, \cref{prop:ZZA_OPTBetweenFullAndFB} implies that the optimal solution has a pending job in the orange segment.
    \caption{Proof of \cref{lem:ZZA_OptJobInFiveFB}}
\end{figure}

\begin{proof}
    Applying \cref{prop:ZZA_FarBehindHasJob} to $i_2,\cdots,i_5$ implies that there exist four jobs $q_2,\cdots, q_5$ such that for every $j\in \pc{2,3,4,5}$ it holds that $\cls{q_j} \in \I{i_j-\sep}{i_j}$ (since the distance between any two of the five classes is at least $2\sep$, these four jobs are distinct).

    If a job $q\in \pc{q_2,\cdots,q_5}$ is a full job, we apply \cref{prop:ZZA_OPTBetweenFullAndFB} to the far-behind class $i_1$ and to $q$ to obtain that the optimal solution has a pending job in the class range $\I{i_1}{\cls{q}}$, which is contained in $[i_1,i_5]$, thus completing the proof.

    Otherwise, assume that $\pc{q_2,\cdots,q_5}$ are all partial jobs.
    Consider the four consecutive partial jobs starting with $q_2$, denoted as $q_2,r_1,r_2,r_3$, such that $\cls{q_2} < \cls{r_1}< \cls{r_2} < \cls{r_3}$.
    It necessarily holds that $\cls{r_3} \le \cls{q_5}$.
    We apply \cref{prop:ZZA_FullBetweenFourPartial} to obtain a full job $q$ such that $\cls{q} \in [\cls{q_2}, \cls{r_3}]$, which is contained in $\I{i_2-\sep}{i_5}$.
    As before, we apply \cref{prop:ZZA_OPTBetweenFullAndFB} to $i_1$ and $q$ which yields that the optimal solution has a pending job in the class range $[i_1,i_5]$.
\end{proof}
The following corollary is immediate from \cref{lem:ZZA_OptJobInFiveFB}.
\begin{cor}
    \label{cor:ZZA_SparseFBBoundedByOPT}
    $\ps{S'} \le 5\sgen*(t) + 4$
\end{cor}

We now return to proving \cref{lem:ZZA_FBBoundedByOPT}.
\begin{proof}[Proof of \cref{lem:ZZA_FBBoundedByOPT}]
    Results immediately from \cref{obs:ZZA_SparsificationProperties,cor:ZZA_SparseFBBoundedByOPT}.
\end{proof}

We can now complete the proof of \cref{lem:ZZA_LocalCompetitiveness}, which implies  \cref{thm:ZZA_Competitiveness}.

\begin{proof}[Proof of \cref{lem:ZZA_LocalCompetitiveness}]
    The lemma results immediately from \cref{lem:ZZA_NumJobsBoundByFB,lem:ZZA_FBBoundedByOPT}.
%
%
%
\end{proof}

    \section{Lower Bound}
    \label{sec:LB}
    \newcommand{\E}{\mathbb{E}}
\newcommand{\D}{\mathcal{D}}
\newcommand{\DD}{{\mathcal{D}'}}
\newcommand{\Dhat}{{\hat{\mathcal{D}}}}

In this section, we show a lower bound for the robust scheduling model.
This lower bound shows that a sublinear dependence on the distortion is impossible in either robust or distortion-oblivious algorithms.


\begin{thm}
    \label{thm:LB}
    For every choice of distortion parameter $\dstr$, every randomized (or deterministic) algorithm is $\Omega(\dstr)$-competitive on inputs with distortion at most $\dstr$.
\end{thm}

The proof of \cref{thm:LB} appears in \cref{sec:LBProof}.

    \section{Discussion and Open Problems}
    \label{sec:Disc}
In this paper, we presented the first distortion-oblivious algorithms for total flow time, which also have a nearly optimal competitive ratio.
Thus, this paper essentially closes the problem of robustness/distortion-obliviousness for total flow time.

It would be interesting to see whether distortion-oblivious algorithms could be designed for other scheduling goals.
A prominent example is \emph{weighted} flow time: while~\cite{DBLP:journals/corr/abs-2103-05604} introduced robust algorithms for this problem, no distortion-oblivious algorithms are known.
One could also consider other goals, such as minimizing mean stretch (ratio of flow time to processing time).
Finally, extending the distortion model to multiple machines and obtaining distortion-oblivious algorithms seems like another natural direction.

    \bibliographystyle{plain}
    \bibliography{bibfile}

    \appendix

    \section{The \texorpdfstring{$\thres$}{DL} Algorithm}
    \label{sec:DLA}
    In this section, we present and analyze the $\thres$ algorithm, a distortion-oblivious algorithm which, for every $\dstr$, is $O(\dstr\log^2\dstr)$-competitive for inputs with distortion $\dstr$ which contain only underestimations.

\subsection{The \texorpdfstring{$\thres$}{DL} Algorithm}

\paragraph{Description of $\thres$.} As in the $\zza$ algorithm, the $\thres$ algorithm maintains a set of partial jobs, which are the only jobs which undergo processing (and there is at most one such partial job per class).

At any point in time, the algorithm processes the minimum-class partial job $q$, unless:
\begin{enumerate}
    \item there exists a pending (full) job of a lower class, and
    \item there exists a third pending job of class less than $\cls{q} + \esep$.
\end{enumerate}
If both conditions hold, the minimum-class full job is marked as partial.
The global parameter $\esep$ is the algorithm's estimate for the parameter $\sep$; this parameter $\sep = \ceil(\log \dstr)+1$ is as defined in \cref{defn:ZZA_Seperator} for $\zza$.
The parameter $\esep$ is updated according to the distortion witnessed by the algorithm (through processing jobs for more than their estimated processing times).

The $\thres$ algorithm is given in \cref{alg:DLA}.
\begin{algorithm}[h]
    \caption{\label{alg:DLA} $\thres$ Algorithm}
    initialize $\esep \gets 2$.

    \While{there exist pending jobs}{
        \If{there is no partial job}{
            Mark the minimum-class pending job as partial.

            \Continue to the next iteration of the loop.
        }

        Let $q$ be the minimum-class partial job in the algorithm.

        \If{there exists a full job of class smaller than $\cls{q}$ \And  there exists another full job of class smaller than $\cls{q} + \esep$}{
            Mark the minimum-class full job as partial.

            \Continue to the next iteration of the loop.
        }

        Process $q$.

        \If{$q$ has been processed for more than $2^i\cdot \ept{q}$ time for some $i$}{
            Set $\esep \gets \max(\esep, i+2)$
        }
    }
\end{algorithm}

\begin{thm}
    \label{thm:DLA_Competitiveness}
    For every $\dstr$, \cref{alg:DLA} is $O(\dstr\log^2\dstr)$-competitive for inputs of distortion $\dstr$ with no overestimations (i.e. $\ddstr=1$).
\end{thm}

\subsection{Analysis}

The analysis of $\thres$ follows the general structure of the analysis of $\zza$.
However, some of the lemmas and propositions require different proofs.

As in the analysis of the $\zza$, the main lemma shows local competitiveness, and immediately implies \cref{thm:DLA_Competitiveness}.

\begin{lem}
    \label{lem:DLA_LocalCompetitiveness}
    Consider an input with distortion at most $\dstr$ which is without overestimations (i.e. $\ddstr = 1$).
    At any time $t$, it holds that $\sgen(t) \le O(\dstr\log^2\dstr)\cdot \sgen*(t)$.
\end{lem}

We henceforth fix a time $t$ towards proving \cref{lem:DLA_LocalCompetitiveness}, and (as in the $\zza$ analysis) assume $\sgen*(t) \ge 1$.

\subsubsection{Bounding \texorpdfstring{$\sgen(t)$}{delta(t)} by Far-Behind Classes}

We use the same notation defined in \cref{defn:ZZA_VolumeAndDeltaDefs}.
In addition, we again define far-behind classes; the following is a restatement of \cref{defn:ZZA_FB} where we note that $\ddstr = 1$.

\begin{defn}[restatement of \cref{defn:ZZA_FB}]
    \label{defn:DLA_FB}
    For every class $i$, we say that $i$ is \emph{far behind} at $t$ if $\Delta\Vgen[\le i](t) \ge 2^i$.

    We also denote by $S$ the set of far-behind classes at $t$.
\end{defn}

The following lemma bounds the number of living jobs in the algorithm by the number of far-behind classes at $t$.
\begin{lem}
    \label{lem:DLA_NumJobsBoundByFB}
    $\sgen(t)  \le O(\dstr\log \dstr) \sgen*(t) + O(\dstr\log \dstr) \cdot \ps{S}$
\end{lem}

We now focus on proving \cref{lem:DLA_NumJobsBoundByFB}.

We define $\sep$ as in \cref{defn:ZZA_Seperator}.
This parameter $\sep$ is what the variable $\esep$ attempts to learn as the algorithm progresses; note that the value of $\esep$ is always at most $\sep$.

\begin{prop}
    \label{prop:DLA_FullBetweenPartial}
    Let $q_1, q_2$ be two partial jobs of classes $i_1, i_2$ respectively, such that $i_1 < i_2$.
    Then there exists a full job in the range $[i_1, i_2 + \sep]$.
\end{prop}
\begin{proof}
    Assume that there is no other partial job in the range $(i_1, i_2)$ (otherwise, replace $q_2$ with this job and continue with the proof).
    It must be that $q_1$ became partial after $q_2$.
    When $q_1$ became partial, $q_2$ was the minimum-class partial job.
    Since $q_2$ was not processed (and $q_1$ was processed instead) there must be another full job $q_3$ (other than $q_2$) at that time of class at most $i_2 + \esep$ for the value of $\esep$ at that time, which is at most $i_2 + \sep$.
    In addition, since $q_1$ was the minimum-class job at the time, the class of $q_3$ was at least $i_1$.
    This completes the proof.
\end{proof}

We use the notation defined in \cref{defn:ZZA_JobNumDefs} to refer to the number of full/partial jobs.
Note that \cref{obs:ZZA_OnlyOnePartialPerClass} for $\zza$ applies to $\thres$ as well; thus, the following corollary holds.
\begin{cor}
    \label{cor:DLA_PartialIsLogTimesFull}
    It holds that $\sgen<p>(t) \le (\sep+2)\sgen<f>(t)$.
\end{cor}

Note that \cref{lem:ZZA_VolumeConversion} for the $\zza$ algorithm is true independently of the algorithm, and thus applies also for $\thres$.
\begin{lem}[restatement of \cref{lem:ZZA_VolumeConversion}]
    \label{lem:DLA_VolumeConversion}
    It holds that
    \[
        \sgen<f>(t)  \le O(\dstr) \cdot \sgen*(t) + O(\dstr) \cdot \ps{S}
    \]
\end{lem}

We can now prove \cref{lem:DLA_NumJobsBoundByFB}.

\begin{proof}[Proof of \cref{lem:DLA_NumJobsBoundByFB}]
    Results immediately from observing that $\sep = O(\log \dstr)$, and applying \cref{cor:DLA_PartialIsLogTimesFull,lem:DLA_VolumeConversion}.
\end{proof}

\subsubsection{Bounding Far-Behind Classes}

We would now like to bound the number of far-behind classes $\ps{S}$.

\begin{lem}
    \label{lem:DLA_FBBoundedByOPT}
    $\ps{S} \le O(\log \dstr)\cdot \sgen*(t)$.
\end{lem}

We perform the same sparsification process as in the $\zza$ analysis to obtain $S'$, and note that \cref{obs:ZZA_SparsificationProperties} applies to $S'$.
We also define $\lastt{i}$ for every $i$, as in \cref{defn:ZZA_LastProcTime}.

Note that the proof of \cref{prop:ZZA_FarBehindHasJob} remains true for $\thres$ as well.
\begin{prop}[restatement of \cref{prop:ZZA_FarBehindHasJob}]
    \label{prop:DLA_FarBehindHasJob}
    If class $i$ is far behind at $t$, then there exists a pending job at $t$ of some class in $\I{i-\sep}{i}$.
\end{prop}

We would now like to prove \cref{prop:ZZA_OPTBetweenFullAndFB} for $\thres$.
\begin{prop}[restatement of \cref{prop:ZZA_OPTBetweenFullAndFB}]
    \label{prop:DLA_OPTBetweenFullAndFB}
    Let $i, i'$ be two classes such that $i' \le i - \sep$, $i'$ is far-behind, and there exists a full job of class $i$ in the algorithm at $t$.
    Then the optimal solution has a job alive in the class range $\I{i'}{i}$.
\end{prop}
\begin{proof}
    \label{go over this proof and remove redundant parts? also define volume with general predicates?}
    First, we require the following claim.

    \textbf{Claim:} if there exists a class $j\in (i', i]$ such that $\Delta \Vgen[\le j](t) \le 0$, then we are done.
    To prove the claim, note that  $\Delta \Vgen [\le i'](t) \ge 2^{i'} > 0$ since $i'$ is far behind, which implies that
    \[
        \Delta \Vgen[\le j, > i'](t) <0
    \]
    and thus the optimal solution must have a job in class range $(i',j] \subseteq (i',i]$ as required.
    This completes the proof of the claim.

    Returning to the proof of the proposition, if the optimal solution has a job of class $i$, then we are done.
    Henceforth assume that it has no such job.

    Let $q$ be the full job of class $i$ in the algorithm.
    Denote by $\lastt{j}$ the last point in time prior to $t$ in which a job of class strictly more than $j$ was processed in the algorithm.
    We now split into cases according to the release time of $q$.

    \textbf{Case 1: $\rlt{q} < \lastt{i-1}$.}
    In this case, at $\lastt{i-1}$ we have that $q$ has already been released, yet a job $r$ of class $\ge i$ is being processed.
    Thus, either the special rule is not being applied, or the special rule \emph{is} being applied in skipping over $q$.
    In either case,  there is no job of class $\le i-1$ at $t$.
    Since from $\lastt{i-1}$ onwards the algorithm only works on jobs of class $\le i-1$, it must be that $\Delta \Vgen[\le i-1](t) \le 0$.
    Since $i-1 \in (i', i]$, the claim above  implies that the proposition holds.

    \textbf{Case 2: $\rlt{q} \ge \lastt{i-\sep}$.}
    In this case, at time $\lastt{i-\sep}$ there exists at most a single job of class at most $i-\sep$, the volume of which is at most $2^{i-\sep} \cdot \dstr \le 2^i$.
    Thus, $\Vgen[\le i-\sep](\lastt{i-\sep}) \le 2^i$.
    During the interval $\I{\lastt{i-\sep}}{t}$, the algorithm only worked on jobs of class at most $i-\sep$, while the optimal solution spent at least $2^i$ time on $q$.

    Thus, $\Delta \Vgen[\le i-\sep](t) \le 0$.
    If $i' = i-\sep$, this is a contradiction to $i'$ being far behind.
    Otherwise, $i-\sep \in \I{i'}{i}$, and thus the above claim implies that the proposition holds.

    \textbf{Case 3: $\rlt{q} \in \IR{\lastt{j}}{\lastt{j-1}}$ for some $j \in \I{i-\sep}{i-1}$.} In this case, consider time $\lastt{j}$, in which a job $r$ of class $>j$ was being processed.
    If this job was processed without the special rule, then there is no living job in the algorithm of class $\le j$ at $\lastt{j}$, which implies that $\Delta \Vgen[\le j](t) \le 0$.
    The claim above would thus imply that the proposition holds.

    Otherwise, the special rule was applied, skipping over a full job $r'$.
    If the job $r'$ is of class more than $j$, we are again done for the same reason.
    Assume therefore that $r'$ is of some class $j'$ which is at most $j$.

    If $\pt{r'} \le 2^{i}$, then observe that $\Vgen[\le j](\lastt{j}) \le 2^i \le \pt{q}$.
    Since the algorithm only works on jobs of class at most $j$ from $\lastt{j}$ onwards, while the optimal solution spends $\pt{q}$ time on completing job $q$, we have that $\Delta \Vgen[\le j](t) \le 0$.
    The above claim would thus imply that the proposition holds;
    henceforth assume that $\pt{r'} > 2^i$.

    If during the interval $\I{\lastt{j}}{t}$ the algorithm spends at most $2^i$ time on job $r'$, then observe that at $\lastt{j}$ the job $r'$ is the only job alive of class $\le j$.
    Thus, the entire interval $\I{\lastt{j}}{t}$ was spent on jobs of class $\le j$ born after $\lastt{j}$, except for at most $2^i$ time units.
    During the same interval, the optimal solution manages to complete the entire job $q$.
    Thus, it holds that $\Delta \Vgen[\le j]<\I{\lastt{j}}{t}>(t) \le 0$, where the time interval in the superscript restricts the considered volume to jobs released in that interval.
    Now, note that the fact that $\pt{r'} \ge 2^i$ implies that $j' > i- \sep \ge i'$.
    Thus, at time $\lastt{j}$ there are no jobs of class $\le i'$.
    This implies that all pending jobs of class $\le i'$ at $t$ were released after $\lastt{j}$, and thus $\Delta \Vgen[\le i']<\I{\lastt{j}}{t}>(t) \ge \Delta \Vgen[\le i'](t) \ge 2^{i'}$, since $i'$ is far behind.
    Thus, $\Delta \Vgen[\le j, > i']<\I{\lastt{j}}{t}>(t) < 0$ which implies that the optimal solution has a living job in some class in $\I{i'}{j} \subseteq \I{i'}{i}$.
    This would complete the proof of the proposition;
    assume henceforth that the algorithm worked on $r'$ for strictly more than $2^i$ time.

    At time $\lastt{j'-1}$, the algorithm works on a job $r''$ of class $\ge j'$.
    But at that point, $r'$ has already been processed for more than $2^i$ time units, and job $q$ of class $i$ has been released; thus, $r''$ is not being processed due to the special rule.
    This implies that there is no job of class $\le j' -1$ at $\lastt{j'-1}$, and thus $\Delta \Vgen[\le j' -1](t) \le 0$.
    Now, note that the fact that $\pt{r'} \ge 2^i$ implies that $j > i- \sep \ge i'$; the claim thus applies and completes the proof of the proposition.
\end{proof}

We can now prove an analogue of \cref{lem:ZZA_OptJobInFiveFB} for the $\thres$ algorithm.
\begin{prop}[analogue of \cref{lem:ZZA_OptJobInFiveFB}]
    Let $i_1, i_2, i_3, i_4$ be four consecutive classes in $S'$.
    Then there exists a job in the optimal solution of class in the range $(i_1, i_4)$.
\end{prop}
\begin{proof}
    First, we claim that there exists a full job $q$ alive in the algorithm such that $\cls{q} \in [i_2 - \sep, i_3+\sep]$

    we apply \cref{prop:DLA_FarBehindHasJob} to $i_2, i_3$ to imply that there exist two pending jobs $q_1, q_2$ in the algorithm at $t$, such that $\cls{q_1} \in \I{i_2 - \sep}{i_2}$ and  $\cls{q_2} \in \I{i_3 - \sep}{i_3}$ (since these class intervals are disjoint, we have that $q_1,q_2$ are distinct).

    If either one of $q_1,q_2$ is full at $t$, we choose $q$ to be that job.
    Otherwise, both $q_1, q_2$ are partial, and we thus apply \cref{prop:DLA_FullBetweenPartial} which implies that the full job $q$ exists such that $\cls{q}\in [\cls{q_1}, \cls{q_2} + \sep] \subseteq [i_2 - \sep, i_3 + \sep]$.

    Observe that $\cls{q} \ge i_2 -\sep \ge i_1 + \sep$; we thus apply \cref{prop:DLA_OPTBetweenFullAndFB} to the far-behind class $i_1$ and the full job $q$, and conclude that the optimal solution has a pending job in the class range $(i_1, \cls{q}) \subseteq (i_1, i_3 + \sep) \subseteq (i_1, i_4)$.

    This completes the proof.
\end{proof}
\begin{cor}
    \label{cor:DLA_FBOptBound}
    $\ps{S'} \le 4\sgen*(t)$
\end{cor}

\begin{proof}[Proof of \cref{lem:DLA_FBBoundedByOPT}]
    Results immediately from \cref{obs:ZZA_SparsificationProperties,cor:DLA_FBOptBound}
\end{proof}

\begin{proof}[Proof of \cref{lem:DLA_LocalCompetitiveness}]
    Results immediately from \cref{lem:DLA_NumJobsBoundByFB,lem:DLA_FBBoundedByOPT}.
\end{proof}

    \section{Lower Bound - Proof of \texorpdfstring{\cref{thm:LB}}{Lower Bound Theorem}}
    \label{sec:LBProof}

%
%
%

In this section, we prove \cref{thm:LB}; the proof takes some ideas from~\cite{Motwani1994}.
For ease of presentation, we first introduce a warm-up \emph{deterministic} lower bound of $\Omega(\dstr)$.
We then show the complete proof of \cref{thm:LB} for randomized algorithms.

\subsection{Warm-up: Deterministic Lower Bound}

We now loosely describe a simple, deterministic lower bound of $\Omega(\dstr)$-competitiveness, before turning to the (somewhat more complex) randomized lower bound.

The adversary releases $n$ jobs (for some large $n$) at time $0$, with estimated processing times of $1$.
The adversary then waits until time $t:=\frac{n\dstr}{2}$.

Denoting by $x_q$ the amount of time spent by the algorithm on job $q$ until $t$, the processing time of $q$ is $\min(x_q+1,\dstr)$.
That is, the algorithm never completes a job by time $t$ unless it spends $\dstr$ time on that job.
Since the processing times of all jobs (which have estimate $1$) is in $[1,\dstr]$, the distortion is indeed at most $\dstr$.

For a job to have less than $1$ unit of time remaining, the algorithm must spend more than $\dstr-1$ units of time on that job.
Thus, from the definition of $t$ it holds that $\sgen(t,1) \ge n-\frac{t}{\dstr-1} \ge \frac{n}{4}$ (recall the definition of $\sgen(t,1)$ from \cref{defn:BC_NumJobsDefn}).


Meanwhile, note that there exist at least $\frac{n}{4}$ jobs $q$ such that $x_q \ge \frac{\dstr}{4}$ (assuming that the algorithm is non-idling).
Denote the set of such jobs by $R$.
The optimal solution could pick a subset $R'\subseteq R$ such that $|R'| = \frac{4n}{\dstr}$ (using $\dstr \ge 16$), and spend the time interval $[0,t]$
as follows:
\begin{itemize}
    \item When the algorithm works on a job not in $R'$, work on that job as well.
    \item When the algorithm works on a job in $R'$, spend this time working on all jobs simultaneously (round robin).
\end{itemize}
Note that the total time devoted by the algorithm to jobs in $R'$ is at least $|R'|\cdot \frac{\dstr}{4} = n$; thus, the optimal solution is able to process every job not in $R'$ for at least one unit of time more than the algorithm (due to the round robin).
But this is enough to finish all jobs except for the jobs of $R'$; thus, $\sgen*(t) \le \frac{4n}{\dstr}$.

We thus have that $\frac{\sgen(t,1)}{\sgen*(t)} = \Omega(\dstr)$, and thus applying the bombardment technique (as stated in \cref{lem:Bombardment}) completes the deterministic lower bound.


\subsection{Randomized Lower Bound: Proof of \texorpdfstring{\cref{thm:LB}}{Lower Bound Theorem}}

We continue to show the randomized lower bound.
We henceforth fix any distortion parameter $\dstr$ and prove \cref{thm:LB} for this distortion parameter.

We prove \cref{thm:LB} using Yao's principle: we describe a distribution on $\dstr$-distorted inputs such that any deterministic algorithm is $\Omega(\dstr)$-competitive against this distribution.
The following proposition reduces the design of such a distribution to designing a distribution in which any deterministic algorithm is bad at some specific time $t$.

\begin{prop}
    \label{prop:LB_BombardmentDistribution}
    If there exists a distribution $\D$ on $\dstr$-distorted inputs and a time $t$ such that for every algorithm $\alg$ it holds that $\frac{\E_\D(\sgen(t,1))}{\E_\D(\sgen*(t))} \ge c$ for some $c$, then there exists a distribution $\Dhat$ over $\dstr$-distorted inputs such that $\frac{\E_\Dhat(\alg)}{\E_\Dhat{\opt}} \ge \Omega(c)$.
\end{prop}
\begin{proof}
    The inputs of $\Dhat$ would behave exactly like the inputs of $\D$ until time $t$ (and would have the same probability).
    From time $t$, all inputs would start a ``bombardment'' sequence, i.e. would release a job $q$ with $\ept{q}=\pt{q}=1$ every time unit for an arbitrarily large number of time units.
    An argument identical to that of \cref{lem:Bombardment} completes the proof.
\end{proof}

It remains to find such a distribution $\D$.

First, we describe the distribution $\DD$, which has unbounded distortion.

\textbf{The distribution $\DD$.} The distribution is defined with respect to the number of jobs $k$. The inputs all consist of releasing $k$ jobs at time $0$, each with predicted processing time $1$.
The real processing times of the jobs are i.i.d. random variables which are picked from the geometric distribution with mean $2$ (i.e. with $p=\frac{1}{2}$).
The adversary then waits for $2(k - k^{3/4})$ time units -- we henceforth define $t:=2(k - k^{3/4})$.

\begin{prop}
    For the distribution $\DD$, it holds that $\frac{\E_\DD(\sgen(t,1))}{\E_\DD(\sgen*(t))} = \Omega(\log k)$.
\end{prop}
\begin{proof}
    First, let's bound the $\E_\DD(\sgen(t,1))$.
    As the jobs have integer processing times, we can assume without loss of generality that the algorithm does not devote fractional time units to jobs.
    Consider any time unit from $0$ to $t$.
    Regardless of the job chosen for processing at that time, the probability that that job will be completed in this time unit is at most $\frac{1}{2}$ (it could be that no job is processed, in which case the probability is $0$).
    Thus, the expected number of jobs completed in $2(k-k^{3/4})$ time units is at most $k-k^{3/4}$.
    Since jobs have integer processing times, jobs that are not completed have at least one time unit of processing remaining, and thus $\E_{\DD}(\sgen(t,1)) \ge k^{3/4}$.

    We now continue to bound the cost of the optimal solution.
    We make the following observations:
    \begin{enumerate}
        \item Denote the total processing time of the $k$ jobs by $Y$.
        Note that $Y$ is the sum of $k$ independent geometric variables with mean $2$, and thus has mean $2k$ and variance $2k$.
        Applying Chebyshev's inequality, $\Pr(Y > 2k+k^{3/4}) \le O(1/\sqrt{k})$.
        \item Defining $b := \frac{\log k}{4}$, the probability that a specific job of the $k$ jobs has processing time more than $b$ is $2^{-b} = k^{-1/4}$.
        Denoting by $B$ the number of jobs with processing time more than $b$, it holds that $\E(B) = k^{3/4}$.
        Moreover, the variance of $B$ is $O(k)$; thus, applying Chebyshev's inequality implies that $\Pr(B<\frac{k^{3/4}}{2}) \le O(1/\sqrt{k})$.
    \end{enumerate}

    Thus, with probability $1-O(1/\sqrt{k})$, it holds that $Y\le k+k^{3/4}$ and $B\ge \frac{k^{3/4}}{2}$.
    Thus, by time $t$ the optimal solution can finish all jobs except for at most $O(\frac{k^{3/4}}{b})$ jobs with volume at least $b$ each.
    Thus, we can bound the expected number of jobs in the optimal solution at $t$ by:
    \[
        \E_{\DD}\pr{\sgen*(t)} \le O\pr{\frac{k^{3/4}}{\log k}} + O\pr{\frac{1}{\sqrt{k}}}\cdot k \le O\pr{\frac{k^{3/4}}{\log k}}
    \]
    which completes the proof of the proposition.
\end{proof}

We can now prove \cref{thm:LB}.

\begin{proof}[Proof of \cref{thm:LB}]
    We construct the distribution $\D$ from the distribution $\DD$ by choosing $k = \floor{2^{\frac{\dstr}{2}}}$ and conditioning on the event $L$ that the processing times of jobs never exceed $\dstr$.
    This new distribution $\D$ thus has a distortion that is bounded by $\dstr$.
    Now, observe that:
    \begin{align}
        \label{eq:LB_RatioBound}
        \frac{\E_\D(\sgen(t,1))}{\E_\D(\sgen*(t))} &= \frac{\Pr(L)\E_\DD(\sgen(t,1) | L)}{\Pr(L)\E_\DD(\sgen*(t) | L)} \ge \frac{\E_\DD(\sgen(t,1)) - \Pr(\bar{L})\E_\DD(\sgen(t,1) | \bar{L})}{\E_\DD(\sgen*(t))}  \\
        &\ge \Omega(\log k) - \frac{\Pr(\bar{L})\E_\DD(\sgen(t,1) | \bar{L})}{\E_\DD(\sgen*(t))} \nonumber
    \end{align}
    Now note that using the union bound on the processing times of jobs, $\Pr(\bar{L}) \le k\cdot 2^{-\dstr} \le 1/k$.
    In addition, denoting by $Y$ the sum of processing times of the $k$ jobs (as before), and noting that $\E_\DD(Y) = 2k$ and $\operatorname{Var}_\DD(Y) = 2k$, we use Chebyshev's inequality to claim that $\Pr(Y\le t) \le O(\frac{1}{k}) \le \frac{1}{2}$, and thus $\E_\DD(\sgen*(t)) \ge \frac{1}{2}$.
    Plugging these observations into \cref{eq:LB_RatioBound} yields that
    \begin{align*}
        \frac{\E_\D(\sgen(t,1))}{\E_\D(\sgen*(t))} &\ge \Omega(\log k) - \frac{\frac{1}{k}\cdot k}{\frac{1}{2}} = \Omega(\log k) - 2 \\
        &= \Omega(\log k) = \Omega(\dstr)
    \end{align*}
    Applying \cref{prop:LB_BombardmentDistribution} to the distribution $\D$ yields a distribution $\Dhat$ with maximum distortion $\dstr$ such that $\frac{\E_{\Dhat}(\alg)}{\E_{\Dhat}(\opt)} = \Omega(\dstr)$, which completes the proof of the theorem.
\end{proof}

    \section{Poor Performance of Existing Algorithms}
    \label{sec:BadCases}
    
In this section, we show that some existing scheduling algorithms are not competitive in our setting.

\begin{defn}
    \label{defn:BC_NumJobsDefn}
    When considering the running of an algorithm on some input, we denote the number of pending jobs in the algorithm at time $t$ by $\sgen(t)$.
    Similarly, we denote the number of pending jobs in the optimal solution at time $t$ by $\sgen*(t)$.
    We also use the notation $\sgen(t,x)$ to denote the number of pending jobs in the algorithm with remaining volume at least $x$.
\end{defn}

The following lemma is a restatement of the standard ``bombardment'' technique in flow-time scheduling.
\begin{lem}
    \label{lem:Bombardment}
    Consider a specific deterministic algorithm.
    If there exists an input for which $\sgen(t,1) \ge c\cdot \sgen*(t)$ at some point in time $t$, then the algorithm is $\Omega(c)$-competitive.
\end{lem}
\begin{proof}
    Suppose such an input $I$ exists.
    Consider the modified input $I'$ which behaves like $I$ until time $t$, but from time $t$ releases a job with processing time $1$ every time unit for $M$ time units.

    The offline solution for this problem would behave like the optimal solution for $I$ until time $t$, but would start working on the stream of jobs of processing time $1$ from $t$ onwards.
    At every time $t'$ during the time interval $[t,t+M]$, the number of pending jobs in the offline solution is at most $\sgen*(t) + 1$.

    Meanwhile, the algorithm has no better option than working on the stream of jobs with processing time $1$, which implies that $\sgen(t') \ge \sgen(t)+1$ for every $t' \in [t,t+M]$.
    Thus, as $M$ tends to $\infty$, the ratio between the algorithm's cost and the offline solution's cost tends to $\frac{\sgen(t)+1}{\sgen*(t)+1} \ge \frac{c}{2}$.
    This completes the proof.
\end{proof}

\subsection{Bad Case for \texorpdfstring{$\sept$}{SEPT}}

The $\sept$ algorithm (shortest estimated processing time) would always choose to process a job from the minimum class of estimated processing time (preferring a partial job if possible).

Even without distortion, this algorithm has an unbounded competitive ratio.
To see this, consider the following input for an arbitrarily large, even $i$:
\begin{enumerate}
    \item For $j$ from $i$ down to $\frac{i}{2}$:
    \begin{enumerate}
        \item Release a job of processing time $2^j + 1$.
        \item Wait $2^j$ time units.
    \end{enumerate}
\end{enumerate}
Denote by $t$ the time in which this input ends.
At $t$, the algorithm $\sept$ would have $\frac{i}{2}+1$ pending jobs of remaining processing time $1$, as it switches to each newly-released job upon its release.
Thus, $\sgen(t,1) \ge \frac{i}{2}$.
However, the optimal solution could have $\sgen*(t) = 1$ in the following way: follow $\sept$ until the job $q$ of class $\frac{i}{2}$ is released, then use the remaining $2^{\frac{i}{2}}$ time units to finish all jobs except $q$ (these jobs require only $\frac{i}{2}$ time to complete).
Using \cref{lem:Bombardment}, and since $i$ is arbitrarily large, the competitive ratio of $\sept$ is unbounded.

\subsection{Bad Case for \texorpdfstring{$\sr$}{SR}}

We now consider the special rule algorithm presented in \cite{DBLP:journals/tcs/BecchettiLMP04}, denoted by $\sr$.
In this algorithm, we again consider the classes of jobs.
The algorithm always works on the lowest-class partial job $q$, until there exist two jobs such that one job is of class at most $\cls{q}$ and the other is of class \emph{strictly less} than $\cls{q}$.
If such jobs exist, the algorithm chooses the one with minimal class and marks it as partial.

We show that this algorithm has an unbounded competitive ratio even for a distortion parameter $\dstr$ which is a moderate constant, specifically $\dstr = 4$.
In fact, we only require underestimations, and so we choose $\ddstr = 1$ and $\udstr = 4$.

Let $i$ be arbitrarily large.
The input for which the algorithm fails is as follows:
\begin{enumerate}
    \item Release a job $q_i$ with estimated processing time $2^i$ and real processing time $2^{i+2}$ (class $i$, distortion of $4$)
    \item For $j$ from $i-1$ down to $0$:
    \begin{enumerate}
        \item Release a job $q_j$ with estimated processing time $2^j$ and real processing time $2^{j+2}$ (class $j$, distortion of $4$).
        \item Release a job $r_j$ with estimated processing time $2^{j+2}$ and real processing time $2^{j+2}$ (class $j+2$, no distortion).
        \item Wait $2^{j+3}$ time units.
    \end{enumerate}
\end{enumerate}
\Cref{fig:BC_SRBadCase} visualizes the bad input for $\sr$.
Initially, $q_i$ arrives and becomes partial in the algorithm (a partial job appears as a red circle).
Immediately afterwards, $q_{i-1}$ and $r_{i-1}$ are released, and are full in the algorithm.
This state is shown in \cref{subfig:BC_SRBadCase1}.

Now, the input waits $8\cdot 2^{i-1} = 2^{i+2}$ time units, during which the algorithm finishes $q_i$, and marks $q_{i-1}$ as partial.
Then, $q_{i-2}$ and $r_{i-2}$ are released.
This state is shown in \cref{subfig:BC_SRBadCase2} (the complete job $q_i$ is shown as a gray circle).

As time progresses, the algorithm reaches the state in \cref{subfig:BC_SRBadCase3}, in which all jobs $\pc{r_j}$ and the job $q_0$ are full and pending.
Thus, $\sgen(t,1) \ge i$.

Meanwhile, the optimal solution could, for every $j$, use the time spent by the algorithm on $q_j$ to finish both $q_{j-1}$ and $r_{i-1}$.
Thus, the only living job in the optimal solution at $t$ would be $q_i$, implying $\sgen*(t) = 1$.
Using \cref{lem:Bombardment} implies that the algorithm has a competitive ratio of $\Omega(i)$, and since $i$ can be arbitrarily large, this competitive ratio is unbounded.

%

\begin{figure*}
    \begin{center}
        \subfloat[][\label{subfig:BC_SRBadCase1}]{
            \includegraphics[width=0.8\columnwidth]{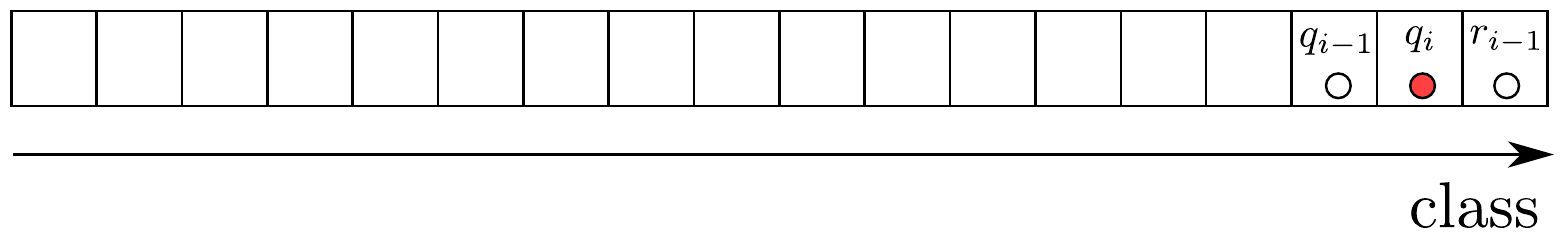}
        }
        \vspace{1em}
        \subfloat[][\label{subfig:BC_SRBadCase2}]{
            \includegraphics[width=0.8\columnwidth]{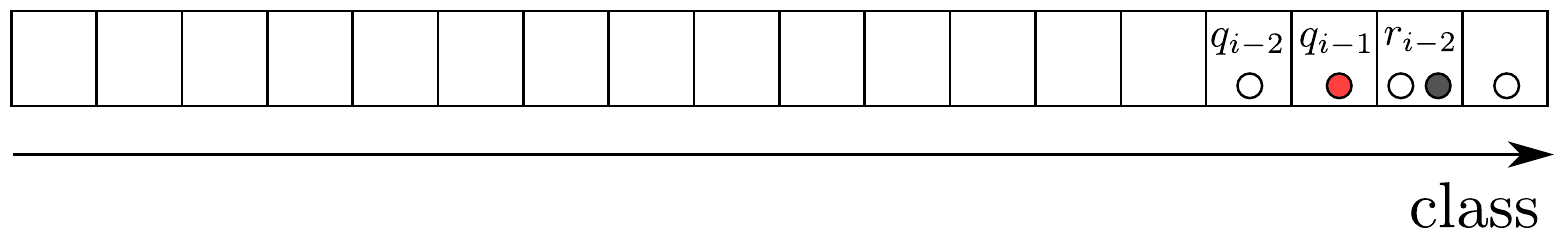}
        }\vspace{1em}
        \subfloat[][\label{subfig:BC_SRBadCase3}]{
            \includegraphics[width=0.8\columnwidth]{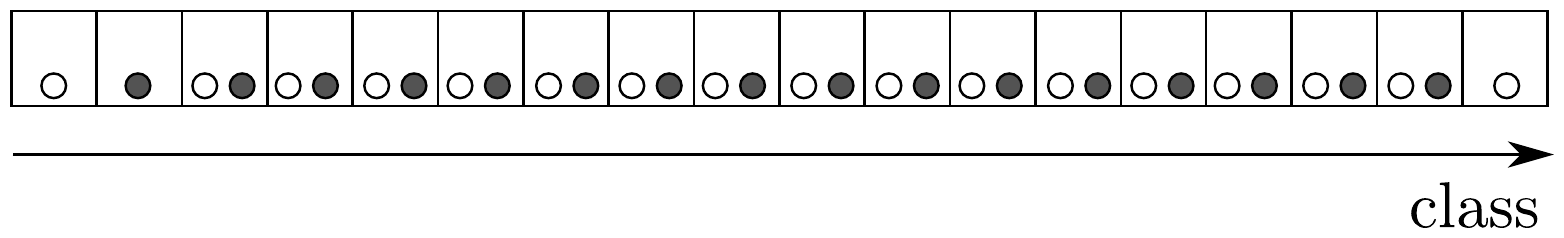}
        }
        \vspace{1em}
    \end{center}
    \caption{\label{fig:BC_SRBadCase} Constant Distortion in the $\sr$ Algorithm of~\cite{DBLP:journals/tcs/BecchettiLMP04}}
\end{figure*}

\subsection{Bad Cases for the Robust Algorithm of~\texorpdfstring{\cite{DBLP:journals/corr/abs-2103-05604}}{Azar et al.}}

For every $\edstr>1$, a $\edstr$-robust, $O(\edstr^2)$ competitive algorithm $\alg_{\edstr}$ was presented in~\cite{DBLP:journals/corr/abs-2103-05604}.
We show two bad cases for the algorithm $\alg_\dstr$.
First, we consider the case in which the distortion $\dstr$ is much smaller than the distortion cap $\edstr$, and show that the algorithm is still $\Omega(\edstr)$-competitive;
In fact, we show this for the case that there is no distortion \emph{at all}, i.e. $\dstr = 1$.
Second, we consider the case in which the distortion $\dstr$ is slightly larger than $\edstr$ (specifically, $\dstr = 4\edstr$), and show that the algorithm has unbounded competitiveness.

\textbf{Bad case 1: $\dstr = 1$ (no distortion).}
Assume $\edstr$ is an even integer for the sake of presentation.
The input consists of releasing $\edstr$ jobs of size $1$ and $2$ jobs of size $\frac{\edstr}{2}$ at time $0$, then waiting for $\edstr$ time units.

\begin{figure}
    \begin{center}
        \includegraphics[width=0.8\columnwidth]{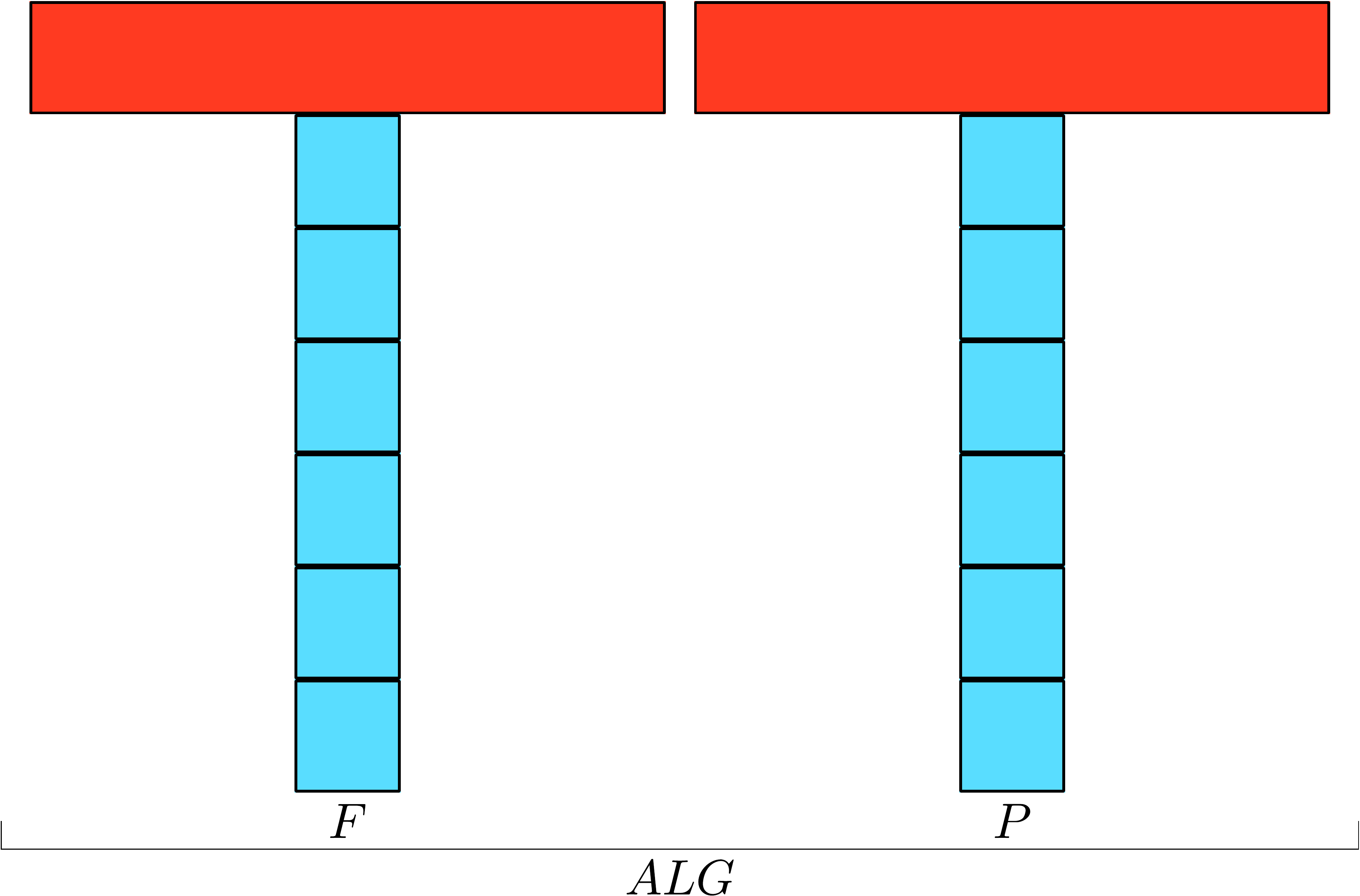}
    \end{center}

    \caption{\label{fig:BC_RotationSmallDistortion} No Distortion in the Algorithm of~\cite{DBLP:journals/corr/abs-2103-05604}}
\end{figure}

The state of the algorithm immediately after the release of the jobs is visualized in \cref{fig:BC_RotationSmallDistortion}.
This visualization follows the description in \cite{DBLP:journals/corr/abs-2103-05604} (i.e. each job is a rectangle whose width is the job's remaining volume).
During the waiting time of $\edstr$ time units, the algorithm would work and finish the two jobs of volume $\frac{\edstr}{2}$, while the optimal solution could finish all jobs of volume $1$.
Thus, at $t=\edstr$ it holds that $\frac{\sgen(t,1)}{\sgen*(t)} =\Omega(\edstr)$, and thus \cref{lem:Bombardment} implies that the algorithm is $\Omega(\edstr)$-competitive even when there is no distortion.

\textbf{Bad case 2: $\dstr = 4\edstr$.}
Assume $\edstr = 2^m$ for some integer $m$.
We construct a somewhat similar adversary to that previously described for $\sr$ with distortion $4\edstr = 2^{m+2}$ and show that the algorithm has unbounded competitive ratio on this input

We choose an arbitrarily large integer $i$.
The adversary performs the following actions:
\begin{enumerate}
    \item Release a job $d$ of arbitrary volume.
    \item Release a job $q_i$ such that $\ept{q_i} = 2^i$, $\pt{q_i}=2^{i+m+2}$.
    \item For $j$ from $i-1$ down to $1$:
    \begin{enumerate}
        \item Release a job $r_j$ such that $\ept{r_j} = \pt{r_j} = 2^{j+m}$ (no distortion).
        \item Release a job $q_j$ such that $\ept{q_j} = 2^j$, $\pt{q_j}=2^{j+m+2}$ (distortion $4\edstr$)
        \item Wait $3\cdot 2^{j+m}$ time units.
    \end{enumerate}
\end{enumerate}

\begin{figure*}
    \begin{center}
        \subfloat[][\label{subfig:BC_RotationLargeDistortion1}]{
            \includegraphics[width=0.8\columnwidth]{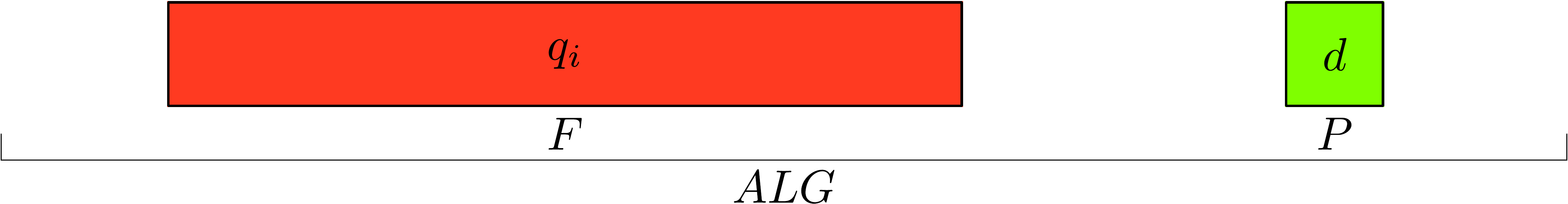}
        }
        \vspace{1em}
        \subfloat[][\label{subfig:BC_RotationLargeDistortion2}]{
            \includegraphics[width=0.8\columnwidth]{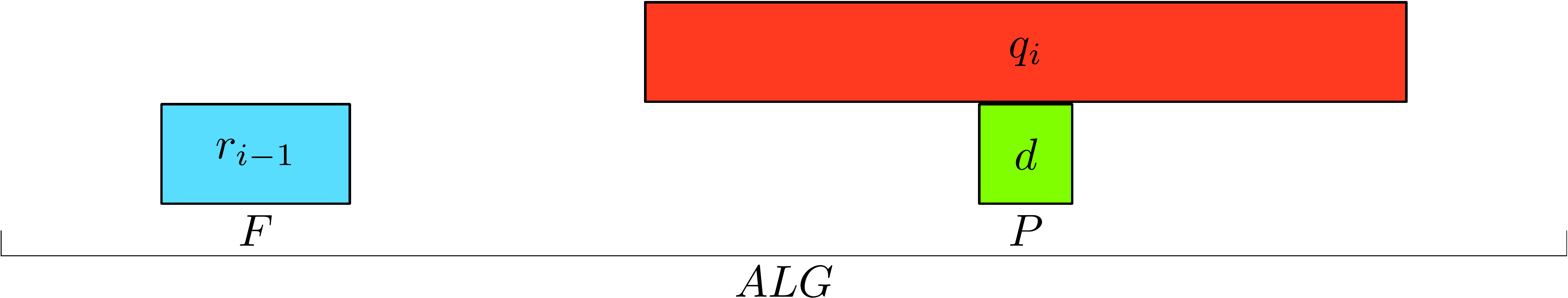}
        }\vspace{1em}
        \subfloat[][\label{subfig:BC_RotationLargeDistortion3}]{
            \includegraphics[width=0.8\columnwidth]{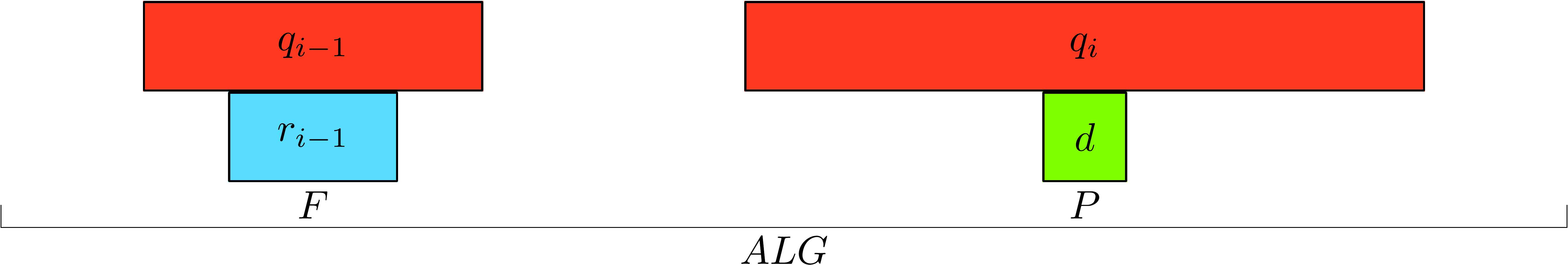}
        }
        \vspace{1em}
        \subfloat[][\label{subfig:BC_RotationLargeDistortion4}]{
            \includegraphics[width=0.45\columnwidth]{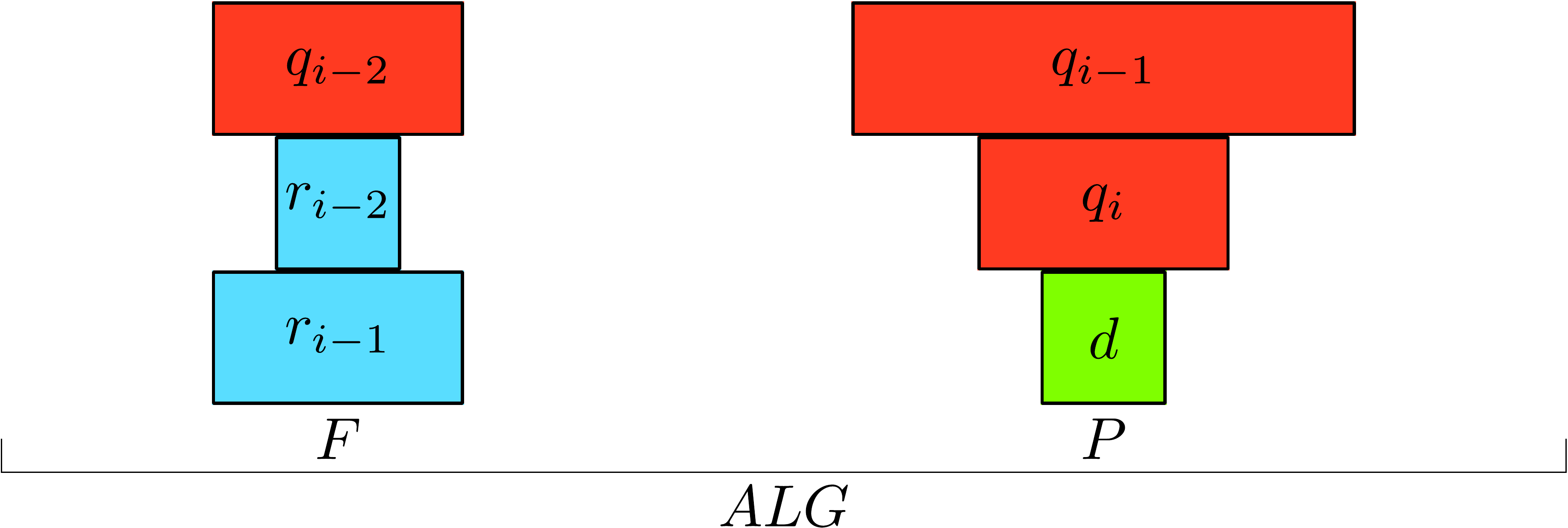}
        }
    \end{center}
    \caption{\label{fig:BC_RotationLargeDistortion} Excessive Distortion in the Algorithm of~\cite{DBLP:journals/corr/abs-2103-05604}}
\end{figure*}
Denote the time in which the adversary ends by $t$.
We consider a visual representation of the algorithm's operation given in \cref{fig:BC_RotationLargeDistortion} (which again follows the description in \cite{DBLP:journals/corr/abs-2103-05604}).
Initially, $d$ is released and immediately moved by the algorithm to the partial bin $P$.
$q_i$ is released immediately afterwards, and is put in the full bin $F$.
\Cref{subfig:BC_RotationLargeDistortion1} shows the state at this point.

Now, $r_{i-1}$ is released into $F$, and swaps with $q_i$ (as its estimate is larger by a factor of $\edstr$).
The job $q_i$ immediately moves to $P$.
The current state is shown in \cref{subfig:BC_RotationLargeDistortion2}.
Now, the job $q_{i-1}$ is released to the top of $F$, as shown in \cref{subfig:BC_RotationLargeDistortion3}.

Now, the adversary waits $3\cdot 2^{i-1+m}$ time, during which the algorithm works on $q_i$ (and doesn't complete it).
Meanwhile, the optimal solution would finish both $q_{i-1}$ and $r_{i-1}$.
Afterwards, the adversary releases $r_{i-2}$ (which swaps with $q_{i-1}$, causing it to move to $P$) and then $q_{i-2}$.
This state is shown in \cref{subfig:BC_RotationLargeDistortion4}.
Now, the adversary waits for $3\cdot 2^{i-2+m}$ time, during which the algorithm works on $q_{i-1}$ and doesn't finish, and the adversary finishes $q_{i-2}$ and $r_{i-2}$.

As this process continues, the algorithm will have all $2i$ jobs alive at $t$ (with at least one using of volume remaining) and thus $\sgen(t,1) \ge 2i$.
Meanwhile, the optimal solution would only have two jobs pending at $t$ ($d$ and $q_i$), and thus $\sgen*(t) \le 2$.
\cref{lem:Bombardment} implies that the algorithm is $\Omega(i)$-competitive.
Since $i$ is arbitrarily large, this implies that the algorithm $\alg_{\edstr}$ has unbounded competitive ratio on inputs with distortion $4\edstr$.

%

\end{document}